\RequirePackage[l2tabu,orthodox]{nag}
\documentclass[11pt,letterpaper]{article} 
\usepackage[notes=false,later=false,camera=false]{dtrt}
\usepackage[utf8]{inputenc}
\usepackage{ stmaryrd }
\usepackage{xspace,enumerate}
\usepackage[T1]{fontenc}
\usepackage[full]{textcomp}
\usepackage[american]{babel}
\usepackage{mathtools}

\usepackage{amsthm}
\usepackage{empheq}
\usepackage{booktabs}
\usepackage{thm-restate}

   \usepackage{hyperref}
\hypersetup{
colorlinks=true,
urlcolor=Cerulean,
linkcolor=NavyBlue,
citecolor=PineGreen,
linktocpage=true,
bookmarksnumbered=true
}

\renewcommand*\backref[1]{\ifx#1\relax \else (pg. #1) \fi}

\usepackage[capitalise,nameinlink]{cleveref}

\crefname{lemma}{Lemma}{Lemmas}
\crefname{fact}{Fact}{Facts}
\crefname{theorem}{Theorem}{Theorems}
\crefname{mtheorem}{Theorem}{Theorems}
\crefname{itheorem}{Theorem}{Theorems}
\crefname{corollary}{Corollary}{Corollaries}
\crefname{claim}{Claim}{Claims}
\crefname{example}{Example}{Examples}
\crefname{algorithm}{Algorithm}{Algorithms}
\crefname{problem}{Problem}{Problems}
\crefname{definition}{Definition}{Definitions}
\crefname{equation}{}{}
\crefname{strategy}{Strategy}{Strategies}

\usepackage{paralist}
\usepackage{turnstile}
\usepackage[framemethod=TikZ]{mdframed}
\mdfsetup{frametitlealignment=\center}
\usepackage{tikz}
\usepackage{caption}
\DeclareCaptionType{Algorithm}
\usepackage{newfloat}

\newtheorem{theorem}{Theorem}[section]
\newtheorem{lemma}[theorem]{Lemma}
\newtheorem*{lemma*}{Lemma}

\newtheorem{proposition}[theorem]{Proposition}
\newtheorem{fact}[theorem]{Fact}

\theoremstyle{definition}
\newtheorem{definition}[theorem]{Definition}
\newtheorem*{definition*}{Definition}
\newtheorem{remark}[theorem]{Remark}

\newtheorem{algorithm}{Algorithm}
\newtheorem{algorithm-thm}[theorem]{Algorithm}

\usepackage[
letterpaper,
top=1.2in,
bottom=1.2in,
left=1in,
right=1in]{geometry}
\usepackage{newpxtext} 
\usepackage{textcomp} 
\usepackage{mathpazo}
\usepackage[scr=rsfso]{mathalfa}
\usepackage{bm} 
\let\mathbb\varmathbb
\usepackage{microtype}

\usepackage{footnotebackref}

\definecolor{petergreen}{rgb}{0, 0.75, 0}

\allowdisplaybreaks
\newcommand{\FormatAuthor}[3]{
\begin{tabular}{c}
#1 \\ {\small\texttt{#2}} \\ {\small #3}
\end{tabular}
}

\usepackage{macros}

\begin{document}

\title{Solving Random Planted CSPs below the $n^{k/2}$ Threshold}

 \author{
 \begin{tabular}{cc}
 \FormatAuthor{Arpon Basu\thanks{Supported by the Moskewicz Venture Forward Graduate Fellowship and NSF CAREER Award \#2047933.}}{arpon.basu@princeton.edu}{Princeton University} &
 \FormatAuthor{Jun-Ting Hsieh\thanks{Supported by NSF CAREER Award \#2047933.}}{juntingh@cs.cmu.edu}{Carnegie Mellon University} \\
  & \\
  \FormatAuthor{Andrew D. Lin\thanks{Supported by a Princeton AI Lab Seed Grant.}}{andrewlin@princeton.edu}{Princeton University} &
  \FormatAuthor{Peter Manohar\thanks{This material is based upon work supported by the National Science Foundation
under Grant No.\ DMS-1926686}}{pmanohar@ias.edu}{The Institute for Advanced Study} 
  \end{tabular}
  }



\maketitle
\vspace{-1em}

\begin{abstract}
    We present a family of algorithms to solve random planted instances of any $k$-ary Boolean constraint satisfaction problem (CSP). A randomly planted instance of a Boolean CSP is generated by (1) choosing an arbitrary planted assignment $x^*$, and then (2) sampling constraints from a particular ``planting distribution'' designed so that $x^*$ will satisfy every constraint. Given an $n$ variable instance of a $k$-ary Boolean CSP with $m$ constraints, our algorithm runs in time $n^{O(\ell)}$ for a choice of a parameter $\ell$, and succeeds in outputting a satisfying assignment if $m \geq O(n) \cdot (n/\ell)^{\frac{k}{2} - 1} \log n$. This generalizes the $\poly(n)$-time algorithm of~\cite{FeldmanPV15}, the case of $\ell = O(1)$, to larger runtimes, and matches the constraint number vs.\ runtime trade-off established for refuting random CSPs by~\cite{RaghavendraRS17}.
    
    Our algorithm is conceptually different from the recent algorithm of~\cite{GuruswamiHKM23}, which gave a $\poly(n)$-time algorithm to solve \emph{semirandom} CSPs with $m \geq \tilde{O}(n^{\frac{k}{2}})$ constraints by exploiting conditions that allow a basic SDP to recover the planted assignment $x^*$ exactly. Instead, we forego certificates of uniqueness and recover $x^*$ in two steps: we first use a degree-$O(\ell)$ Sum-of-Squares SDP to find some $\widehat{x}$ that is $o(1)$-close to $x^*$, and then we use a second rounding procedure to recover $x^*$ from $\widehat{x}$.
\end{abstract}


\thispagestyle{empty}
\setcounter{page}{0}

\clearpage
 \microtypesetup{protrusion=false}
  \tableofcontents{}
  \microtypesetup{protrusion=true}
\thispagestyle{empty}
\setcounter{page}{0}

\clearpage

\pagestyle{plain}
\setcounter{page}{1}

\section{Introduction}
Constraint satisfaction problems (CSPs) such as $k$-SAT are a foundational class of computational problems that are notoriously hard in the worst case. After several decades, results on hardness of approximation show that, assuming $\mathsf{P} \ne \mathsf{NP}$, no polynomial-time algorithm can perform much better than a naive algorithm that simply returns a random assignment when the instance is sparse: namely, it has $m = O(n)$ constraints where $n$ is the number of variables~\cite{Hastad01}.  While there are efficient algorithms for ``maximally dense'' instances with $m \geq \Omega(n^k)$ constraints~\cite{AroraKK95}, it is not possible to outperform a random assignment if the instance has $m = o(n^k)$ constraints~\cite{FotakisLP16}, assuming the Exponential-Time Hypothesis (\textsf{ETH}) \cite{ImpagliazzoP01}. More generally, under \textsf{ETH}, it is not possible for an algorithm running in time $2^{n^{\delta}}$ to outperform a random assignment if the instance has $m = o(n^{k - \delta})$ constraints~\cite{FotakisLP16}. 

Given such strong hardness results, it is natural to study the average-case \emph{search} problem of solving \emph{random} CSPs, where it might be possible to design substantially better algorithms than in the worst case. Although a fully random CSP (i.e., each constraint is drawn uniformly at random from all possible constraints) with $m \gg n$ constraints is unsatisfiable with high probability, one can define fairly natural models of random \emph{planted} CSPs, where the instance is sampled from a ``highly random'' distribution but is still guaranteed to have a particular satisfying assignment $x^*$.\footnote{Intuitively, one should view this distribution as first choosing an assignment $x^*$ to be the planted assignment, and then drawing $m$ constraints at random conditioned on each one being satisfied by $x^*$. As it turns out, using such a naive process actually makes recovering $x^*$ rather straightforward; e.g., in the case of $k$-SAT, the aforementioned distribution will have the negation signs associated to any variable $i\in[n]$ be biased towards the sign of $x^*_i$, and thus a simple majority ``poll'' for each variable reveals the value of $x^*$ with high probability. Because of this, one has to choose the random constraints a bit more carefully. See \cref{def:randomplantedcsp} for the formal definition of the model.}
This task of solving a random planted CSP has received much attention in prior works~\cite{BarthelHL02, JiaMS07, BogdanovQ09, CojaCF10, FeldmanPV15}, with the work of~\cite{FeldmanPV15} giving an algorithm that runs in time $\poly(n)$ and succeeds in solving the random planted CSP if $m \geq \tilde{O}(n^{k/2})$, which is far lower than the $m = o(n^{k-1})$ threshold for which $2^{n^{1 - o(1)}}$-time hardness is known under \textsf{ETH}~\cite{FotakisLP16}.
Besides hoping to understand the complexity of CSPs better, many of these works were also motivated by the connection to Goldreich's pseudorandom generator~\cite{Goldreich00}, the security of which boils down to distinguishing a random planted CSP from a fully random CSP. For example, the work of~\cite{FeldmanPV15} implies that one can efficiently invert Goldreich's PRG if the PRG is $k$-local and has stretch larger than $\tilde{O}(n^{k/2})$.

In addition to the above average-case search problem, there is another fairly natural average-case problem for CSPs that has also received a lot of attention: \emph{refutation}. In this problem, one is given a fully random CSP (which is unsatisfiable with high probability) and the algorithmic goal is to find a certificate of unsatisfiability. 
The refutation problem has received considerable attention over the past two decades~\cite{GoerdtL03, CojaGL07, AllenOW15, BarakM16, RaghavendraRS17}, culminating in the work of~\cite{RaghavendraRS17}, which gives an algorithm that, for a choice of a parameter $\ell$, runs in $n^{O(\ell)}$ time and succeeds in refuting a random CSP with high probability provided that it has at least $m \geq \tilde{O}(n) \cdot (n/\ell)^{k/2 - 1}$ constraints.
That is, they establish a trade-off between the runtime and number of constraints required for refutation: as the ``runtime exponent'' $\ell$ varies from $O(1)$ to $n$, the number of constraints required varies from $\tilde{O}(n^{k/2})$ to $\tilde{O}(n)$. 
In particular, for the choice of $\ell = O(1)$, this algorithm runs in $\poly(n)$ time and succeeds in refuting a random instance if it has at least $m \geq \tilde{O}(n^{k/2})$ constraints, the same threshold as in~\cite{FeldmanPV15}. This runtime vs.\ number of constraints trade-off established by~\cite{RaghavendraRS17} is conjectured to be optimal, with evidence coming in the form of lower bounds against restricted computational models such as the Sum-of-Squares (SoS) SDP hierarchy~\cite{Feige02, BenabbasGMT12, OdonnellW14, MoriW16, BarakCK15, KothariMOW17, FeldmanPV18}.

In light of the runtime vs.\ number of constraints trade-off established by~\cite{RaghavendraRS17} for the refutation problem, one may wonder if it is possible to get an algorithm with similar performance for the search problem, i.e., for random planted CSPs. While there has been a flurry of recent work on understanding refutation and search for the ``harder'' semirandom and smoothed CSPs~\cite{AbascalGK21,GuruswamiKM22,HsiehKM23,GuruswamiHKM23}, there is still no known algorithm for solving random planted CSPs at the same runtime vs.\ number of constraints trade-off established by~\cite{RaghavendraRS17} for the refutation problem --- in time $n^{O(\ell)}$ with $\tilde{O}(n) \cdot (n/\ell)^{k/2 - 1}$ constraints. While there has been some recent work on designing superpolynomial-time algorithms by~\cite{ChenSZ25}, the runtime vs.\ number of constraints trade-off obtained is far from the conjectured optimal one.\footnote{For example, the algorithm of~\cite{ChenSZ25} always requires at least $\Omega(n^2)$ constraints, regardless of the runtime.}

\subsection{Our results}
\label{sec:results}
As our main result, we give, for any choice of a parameter $\ell$, an algorithm running in time $n^{O(\ell)}$ that succeeds in outputting a satisfying assignment to a random planted $k$-ary CSP with high probability provided that the CSP has at least $m \geq \tilde{O}(n) \cdot (n/\ell)^{k/2 - 1}$ constraints. That is, we resolve the aforementioned gap and give an algorithm with the ``correct'' trade-off matching the trade-off that is known for the refutation case. Similar to~\cite{FeldmanPV15}, we can also achieve a much better dependence on $m$ by replacing $k$ with $r$ in the above expression in the case that the $k$-ary CSP has ``complexity'' $r$ rather than $k$.

Before we formally state our results, we will first formally define the random planted CSP model, as done in~\cite{FeldmanPV15}.

\begin{definition}[$k$-ary Boolean CSPs]
A CSP instance $\Psi$ with a $k$-ary predicate $P \colon \Fits^k \to \Bits$ is a set of $m$ constraints on variables $x_1,\dots,x_n$ of the form $P(\literalneg(C)_1 x_{i_1}, \literalneg(C)_2 x_{i_2}, \ldots, \literalneg(C)_k x_{i_k}) = 1$, where $C = (i_1, \ldots, i_k) \in [n]^k$ ranges over a collection $\cH$ of \emph{scopes}\footnote{We additionally allow $\cH$ to be a multiset, i.e., that multiple clauses can contain the same ordered tuple of variables.} (a.k.a.\ clause structure) of $k$-tuples of $n$ variables and $\literalneg(C) \in \Fits^k$ are ``literal negations'', one for each $C$ in $\cH$. We let $\val_{\Psi}(x)$ denote the fraction of constraints satisfied by an assignment $x \in \Fits^n$, and we define the \emph{value} of $\Psi$, $\val(\Psi)$, to be $\max_{x \in \Fits^n} \val_{\Psi}(x)$.
\end{definition}

\begin{definition}[Random planted $k$-ary Boolean CSPs]
\label{def:randomplantedcsp}
Let $P \colon \Fits^k \to \Bits$ be a predicate. We say that a distribution $\cQ$ over $\Fits^k$ is a \emph{planting distribution for $P$} if $\Pr_{y \sim \cQ}[P(y) = 1] = 1$.

We say that an instance $\Psi$ with predicate $P$ is a \emph{random planted instance} with \emph{planting distribution} $\cQ$ if it is sampled from a distribution $\Psi(x^*, m, \cQ)$ where
\begin{enumerate}[(1)]
    \item The planted assignment $x^* \in \Fits^n$ is arbitrary;
    \item the scopes $\cH \subseteq [n]^k$ is a multiset of size $m$ sampled by choosing $m$ elements of $[n]^k$ uniformly at random with replacement;
    \item for each $C = (i_1, \dots, i_k)\in \cH$, the literal negations $\literalneg(C)$ are sampled by $\literalneg(C) \sim \cQ(\literalneg(C) \odot (x^*_{i_1}, \dots, x^*_{i_k}))$, where ``$\odot$'' denotes the element-wise product of two vectors.
    That is, $\Pr[\literalneg(C) = y] = \cQ(y \odot (x^*_{i_1}, \dots, x^*_{i_k}))$ for each $y \in \Fits^k$.
    Then, add the constraint 
    \[P(\literalneg(C)_1 x_{i_1}, \literalneg(C)_2 x_{i_2}, \ldots, \literalneg(C)_k x_{i_k}) = 1\] 
    to the instance $\Psi$.
\end{enumerate}
Because $\cQ$ is supported only on satisfying assignments to $P$, it follows that if $\Psi \sim \Psi(x^*, m, \cQ)$, then $x^*$ satisfies $\Psi$ with probability $1$.
\end{definition}

Before we state our main theorem, we define the notion of \emph{distribution complexity}, as defined in~\cite{FeldmanPV15}.
\begin{definition}[Distribution Complexity]
\label{def:distcomp}
    Let $P:\Fits^k\to\{0, 1\}$ be a predicate, and let $\cQ:\Fits^k\to[0, 1]$ be a planting distribution supported on $P^{-1}(1)$. The \emph{distribution complexity} of $\cQ$ is defined to be the smallest integer $r\geq 1$ for which there exists a set $S\seq[k]$ of size $r$ such that $|\widehat{\cQ}(S)|\geq 4^{-k}$,\footnote{The definition in~\cite{FeldmanPV15} differs slightly here and only requires that $|\widehat{\cQ}(S)| > 0$. This is because they implicitly assume that $\cQ$ is constant with respect to $n$, so that if $|\widehat{\cQ}(S)| > 0$ and $k = O(1)$, then $|\widehat{\cQ}(S)| \geq \Omega(1)$. We do not make this assumption on $\cQ$.} where for any set $T\seq[k]$, the Fourier coefficient $\widehat{\cQ}(T)$ is defined as $\widehat{\cQ}(T):= 2^{-k}\sum_{y\in\Fits^k}\cQ(y)\prod_{j\in T}y_j$. In case $\max_{\emptyset\neq S\seq[k]}|\widehat{\cQ}(S)| < 4^{-k}$, set the distribution complexity of $\cQ$ to be $r = 1$.
\end{definition}
Note that if $P$ is a non-trivial predicate, i.e.\ $P^{-1}(1)\subsetneq\Fits^k$, then $\supp(\cQ)\neq \Fits^k$ since $\supp(\cQ)\seq P^{-1}(1)$, and thus one can show (see \cref{largestfouriercoefflowerbound}) that $\max_{\emptyset\neq S\seq[k]}|\widehat{\cQ}(S)| \geq 4^{-k}$.

We now state our main theorem, which gives an algorithm to solve a random planted CSP.

\begin{restatable}[Algorithm for Random Planted CSPs]{mtheorem}{cspalg}
\label{mthm:csp}
    Let $m, n$ be positive integers, and let $2\leq k\leq c\log n$ for some small enough constant $c > 0$. Let $\ell$ be a parameter such that $2k \leq \ell \leq n/8$. There is an algorithm that takes as input a $k$-ary Boolean CSP $\Psi$ and in time $n^{O(\ell)}$ outputs an assignment $x \in \Fits^n$ with the following guarantee: if $\Psi \sim \Psi(x^*, m, \cQ)$ where $\cQ$ is a planting distribution with distribution complexity $r$ and
    \begin{equation*}
    m\geq m_0:= 2^{O(k)}\cdot n\log n \cdot \max\left\{\left(\frac{n}{\ell}\right)^{\frac{r}{2} - 1}, 1\right\}  \mcom
    \end{equation*}
    then with probability $\geq 1 - 1/\poly(n)$ over the draw of the random CSP instance $\Psi$, it holds that the output $x$ of the algorithm is a satisfying assignment to $\Psi$. Furthermore, the algorithm can output a list of $\leq 2^{k + 2}$ assignments such that at least one of them is $x^*$.
\end{restatable}
We note that in \cref{mthm:csp}, the first term in the maximum is larger unless $r = 1$. 

\cref{mthm:csp} fills the remaining gap in our knowledge for the search problem of solving random planted CSPs by giving an algorithm with ``\cite{RaghavendraRS17}-like'' guarantees, i.e., an algorithm that matches the conjectured ``correct'' trade-off between runtime and number of constraints. We note that there is a similar gap in our knowledge for the case of \emph{semirandom} CSPs, a model where the scopes $\cH$ are additionally allowed to be worst-case. For semirandom CSPs, the works of \cite{AbascalGK21,GuruswamiKM22,HsiehKM23} give a ``\cite{RaghavendraRS17}-like'' refutation algorithm, and the work of~\cite{GuruswamiHKM23} gives a $\poly(n)$-time algorithm to solve semirandom planted CSPs with $\tilde{O}(n^{k/2})$ constraints. However, no ``\cite{RaghavendraRS17}-like'' search algorithm is known, and we leave this problem as an intriguing open question for future work. As the semirandom case is only harder than the random case handled in \cref{mthm:csp}, \cref{mthm:csp} is a necessary first step towards obtaining an algorithm for the semirandom case.

\cref{mthm:csp} differs slightly from~\cite{FeldmanPV15}, which gives an algorithm to recover $x^*$ exactly (up to a global sign), rather than output a short list of candidate assignments, of which one is guaranteed to be $x^*$ and thus a satisfying one. This is because the order of quantifiers in the result of~\cite{FeldmanPV15} is different: they allow their algorithm to depend on $\cQ$, i.e., it ``knows'' the planting distribution $\cQ$ in advance. If we make the same assumption, then our algorithm does not need to output a list of candidate assignments, and can instead also recover $x^*$ exactly (up to a global sign) as well.

In \cref{mthm:csp}, we assume that $\cQ$ is a planting distribution so that it is supported only on satisfying assignments to $P$. However, if we assume that it only outputs a satisfying assignment to $P$ with probability, say, $1 - \delta$ for a constant $\delta$, so that $\val_{\Psi}(x^*) \geq 1 - \delta - o(1)$ with high probability for $\Psi \sim \Psi(x^*, m, \cQ)$,\footnote{$\val_{\Psi}(x^*)$ is $1 - \delta$ in expectation, and by a Chernoff bound is at least $1 - \delta - o(1)$ with high probability.}
then the algorithm still succeeds in outputting an assignment that is ``just as good'' as $x^*$, i.e., it has value at least $1 - \delta - o(1)$ as well.

\cref{mthm:csp} additionally gives an algorithm to invert Goldreich's PRG~\cite{Goldreich00} and its related variants~\cite{Applebaum16} in time $n^{O(\ell)}$ when the stretch of the PRG is at least $2^{O(k)}\cdot n \cdot \left(\frac{n}{\ell}\right)^{\frac{k}{2} - 1} \cdot \log n$ and $k$ is the ``locality'' of the PRG.

\parhead{Noisy $k$-XOR / $k$-sparse LPN.}
We prove \cref{mthm:csp} using the reduction of~\cite{FeldmanPV15,GuruswamiHKM23} to the special case of ``noisy $k$-XOR'', also referred to as $k$-sparse Learning Parity with Noise (LPN). This problem was first introduced in~\cite{Alekhnovich03}, and has found many applications in cryptography such as, e.g.,~\cite{ApplebaumBW10}. Below, we define the ``search'' or ``learning'' variant of the problem.

\begin{definition}[Noisy $k$-XOR / $k$-sparse LPN instance]
We define the distribution $\LPN_k(x^*, m, \eps)$ over $k$-XOR instances as follows:
\begin{enumerate}[(1)]
\item We sample $\cH$ by choosing $m$ tuples $C \in [n]^k$ with replacement;
\item For each $C \in \cH$, we set the ``right-hand side'' of the equation to be $b_C = \prod_{i \in C}x^*_i$ with probability $1/2 + \eps$ and $b_C = -\prod_{i \in C}x^*_i$ with probability $1/2 - \eps$, independently.
\end{enumerate}
We thus obtain a family of equations $\{\prod_{i\in C}x_i = b_C\}_{C\in\cH}$ over $\Fits^n$, which we refer to as $(\cH, \{b_C\}_{C\in\cH})$. We also refer to equations where $b_C = -\prod_{i \in C}x^*_i$ as \emph{corrupted} or \emph{noisy}.
\end{definition}

The key technical component of the proof of \cref{mthm:csp} is the following theorem, which gives a new algorithm to recover the planted assignment $x^*$ from a noisy $k$-XOR instance.
\begin{restatable}[Algorithm for Noisy $k$-XOR]{mtheorem}{noisykxor}
\label{mthm:noisykxor}
    Let $m,n$ be positive integers and let $1\leq k\leq o(n^{1/4})$. Let $\ell$ be a parameter such that $2k \leq \ell \leq n/8$. There is an algorithm that takes as input a $k$-XOR instance and in time $n^{O(\ell)}$ outputs an assignment $x \in \Fits^n$ with the following guarantee: if the input $(\cH, \{b_C\}_{C\in\cH})$ is drawn from $\LPN_k(x^*, m, \eps)$, where $\eps\in(0, \frac{1}{2}]$ is such that $\eps^2/k\geq\Omega(n^{-1/2})$ 
    and 
    \begin{align}
    \label{eq:mlowerboundmaintheorem}
        m\geq m_0:= 2^{O(k)}\cdot n \log n \cdot \max\left\{\frac{1}{\eps^{11}} \cdot \left(\frac{n}{\ell}\right)^{\frac{k}{2} - 1} , \frac{1}{\eps^2}\right\}\mcom
    \end{align}
     then with probability $\geq 1 - 1/\poly(n)$ over the draw of $(\cH, \{b_C\}_{C\in\cH})$, it holds that the output $x$ satisfies $x = x^*$ if $k$ is odd, and either $x = x^*$ or $x = -x^*$ if $k$ is even.
\end{restatable}
We note that as in \cref{mthm:csp}, the first term in \cref{eq:mlowerboundmaintheorem} is larger unless $k = 1$.

\cref{mthm:noisykxor} gives a new algorithm for the search (a.k.a.\ ``learning'') variant of $k$-sparse LPN problem in the ``high noise regime'' where the noise is $1/2 - \eps$. To the best of our knowledge, prior algorithms in this high noise regime were either for the distinguishing variant of the problem~\cite{AllenOW15,RaghavendraRS17,GuruswamiKM22,HsiehKM23}, the search problem in the polynomial-time regime~\cite{FeldmanPV15,GuruswamiHKM23}, or are far from obtaining the conjectured optimal ``runtime vs.\ constraints'' trade-off~\cite{ChenSZ25}.

Unlike in \cref{mthm:csp}, in \cref{mthm:noisykxor} our algorithm recovers the planted assignment $x^*$ \emph{exactly}: when $k$ is odd, it outputs $x^*$, and when $k$ is even, it outputs one of $\pm x^*$, which is the best possible since $\LPN_k(x^*, m, \eps)$ and $\LPN_k(-x^*, m, \eps)$ are identical distributions. The reason we do not recover $x^*$ exactly in \cref{mthm:csp} is because the reduction from the general CSP case in \cref{mthm:csp} to the noisy $k$-XOR case in \cref{mthm:noisykxor} produces $2^{k+1} - 2$ candidate noisy $t$-XOR instances for $1 \leq t \leq k$, and then each noisy $t$-XOR instance outputs at most $2$ candidates for $x^*$ because of the global sign. If we use the convention of~\cite{FeldmanPV15} and allow the algorithm in \cref{mthm:csp} to ``know'' the planting distribution $\cQ$ in advance, then it ``knows'' which noisy $t$-XOR instance it should use to recover $x^*$, which allows it to recover $x^*$ exactly (up to a global sign) as done in~\cite{FeldmanPV15}.

\section{Proof Overview}
\label{sec:proofoverview}
In this section, we give an overview of the proofs of our two main theorems, \cref{mthm:csp,mthm:noisykxor}. As \cref{mthm:csp} follows from \cref{mthm:noisykxor} combined with the standard reduction from case of general CSPs to $k$-XOR~\cite{FeldmanPV15,GuruswamiHKM23}, we will focus on \cref{mthm:noisykxor} in this section, which is the case of noisy $k$-XOR.

To begin, in \cref{sec:GHKM23} we will briefly recall the overall approach of~\cite{GuruswamiHKM23} and explain why it fails to generalize to the case of super-constant $\ell$. Then, in \Cref{sec:overview-rounding} we will explain our approach and how it avoids the previous technical barrier.

\subsection{Certificates of uniqueness: the algorithm of \texorpdfstring{\cite{GuruswamiHKM23}}{GHKM23}}
\label{sec:GHKM23}

The key idea in the algorithm of~\cite{GuruswamiHKM23} is to understand conditions under which a simple SDP is able to recover the planted assignment $x^*$ exactly. For simplicity, let us briefly explain how this works for the case of $2$-XOR, before explaining how this generalizes to larger $k$.

For the case of noisy $2$-XOR, we are given equations of the form $x_i x_j = b_{ij}$, where $b_{ij}$ is equal to $x^*_i x^*_j$ with probability $1/2 + \eps$. We can represent the pairs $\{i,j\}$ for which we are given equations as a graph $G$, and we can view it as having a signed adjacency matrix $A$, where $A_{i,j} = b_{ij}$ if there is an equation for the pair $\{i,j\}$, and otherwise $A_{i,j} = 0$. The argument of~\cite{GuruswamiHKM23} shows that if the (unsigned) graph $G$ is an expander with \emph{spectral gap} $\lambda$\footnote{That is, $\lambda$ is the second smallest eigenvalue of the normalized Laplacian of $G$, where the smallest eigenvalue of the normalized Laplacian is always $0$.} and $m \geq n \log n \cdot \poly(1/\lambda)$, then with high probability over the noise, the basic SDP is able to \emph{certify} that $x^*$ is (up to a sign) the unique globally optimal assignment. This implies that the basic SDP, which computes the matrix $X \in \R^{n \times n}$ that maximizes $\sum_{\{i,j\} \in G} X_{ij} b_{ij}$ subject to the conditions $X \succeq 0$ and the diagonal of $X$ is all $1$'s, has an optimal solution of $X = x^* (x^*)^{\top}$, and from this one can easily recover $x^*$ up to a sign. This is the best possible when $k$ is even, as the sign of $x^*$ does not matter. Finally, as $G$ is random in the case of noisy $2$-XOR, it is also an expander, and this finishes the proof.

If $k$ is even but larger than $2$, a similar idea still works. One can form the matrix $A$ indexed by sets $S$ of size $k/2$, where $A(S,T) = b_{S \cup T}$ if $S \cup T = C \in \cH$ is a constraint with right-hand side $b_{C}$. Since the hypergraph $\cH$ is random, this underlying graph is again an expander, and so a similar proof then shows that if $m \geq \tilde{O}(N) = \tilde{O}(n^{k/2})$ where $N = \binom{n}{k/2}$ is the number of vertices in the new graph, then one can recover $\prod_{i \in C} x^*_i$ for each $C \subseteq [n]$ of size $k$. From this, one immediately can recover $x^*$ up to a global sign as well.

\parhead{The barrier to the approach of~\cite{GuruswamiHKM23}: weak expansion of Kikuchi graphs.} The above argument shows how to give an algorithm for the ``polynomial-time case'' of $\ell = O(1)$. Let us now explain why this approach fails when $\ell$ becomes super-constant. For general $\ell$, the natural graph to consider is the ``level $\ell$ Kikuchi graph'' of~\cite{WeinAM19}, which is defined below.
\begin{definition}
\label{def:evenkikuchi}
Let $k/2 \leq \ell \leq n$ be an integer, and let $\cH$ be a collection of subsets of $n$ of size exactly $k$. The Kikuchi graph $G(\cH, \ell)$ is defined as follows. The vertices of the graph are the subsets $S \subseteq [n]$ of size exactly $\ell$, and there is an edge $(S,T)$ if the symmetric difference of $S$ and $T$ is equal to some $C \in \cH$.
\end{definition}
The Kikuchi graph in \cref{def:evenkikuchi} was originally introduced by~\cite{WeinAM19} to give a simplified proof of the results of~\cite{RaghavendraRS17}, and has since been very influential in extending algorithms for average-case $k$-XOR instances to the superpolynomial-time and semirandom settings (see~\cite{GuruswamiKM22,HsiehKM23}).
In light of this, one might expect to be able to simply use this graph to extend the algorithm of~\cite{GuruswamiHKM23} to prove \cref{mthm:noisykxor}. However, there is an immediate issue one encounters, which is that the Kikuchi graph is \emph{not} a good expander when $\ell$ is super-constant. Even if we let $\cH$ be the ``complete'' hypergraph, i.e., it contains all sets $C$ of size $k$, then the Kikuchi graph is simply a Johnson graph. The eigenvalues and eigenvectors of the Johnson graph are well-understood (see~\cite[Lemma~A.3]{WeinAM19}), and in this setting of parameters we will have spectral gap $\lambda \sim 1/\ell$. Now, because the constraint threshold of~\cite{GuruswamiHKM23} loses a $\poly(1/\lambda)$ factor, this means that if we rely on expansion here, the resulting constraint threshold we obtain will be $m_0 \cdot \poly(\ell)$ instead of $m_0$, where $m_0 = \tilde{O}(n) \left(\frac{n}{\ell}\right)^{\frac{k}{2} - 1}$ is the ``correct'' threshold. When $\ell = O(1)$ as in the case of~\cite{GuruswamiHKM23}, this extra $\poly(\ell)$ factor is irrelevant and one obtains a threshold of $\tilde{O}(n^{k/2})$, but in the more general setting of \cref{mthm:noisykxor}, this is a significant loss. Furthermore, we do not see a way to overcome this barrier, as the complete hypergraph is in some sense the ``best case scenario'' for the algorithm.

\subsection{A two part rounding procedure: our approach}
\label{sec:overview-rounding}

The  problem in the aforementioned approach arises because we are trying to find an algorithm that recovers $x^*$ in ``one shot'' by certifying that $x^*$ is the globally optimal assignment. Instead, to prove \cref{mthm:noisykxor}, we will depart from this framework and find $x^*$ in two steps. In the first step, we will be content with finding an assignment $\widehat{x}$ that is $o(1)$-close to $x^*$ in normalized Hamming distance. Then, we will explain how, given $\widehat{x}$, we can then round it to recover $x^*$ exactly.

\parhead{Finding an approximate solution.} Our algorithm for finding an approximate solution uses the Sum-of-Squares hierarchy; we refer the reader to \cref{sec:sos} for the relevant background. As notation, we will let $\cH$ denote the hypergraph of the noisy $k$-XOR instance. We will think of each right-hand side $b_C$ as being sampled via a two step process. First, with probability $2 \eps$, we set $b_C = x^*_C$, and otherwise we set $b_C = \sigma_C$, where $\sigma_C$ is a Rademacher random variable, i.e., it is uniform over $\Fits$. We let $\cH_{2 \eps}$ denote the set of $C$'s where the sampling procedure sets $b_C = x^*_C$ in the first step, so that for $C \in \cH \setminus \cH_{2 \eps}$, we have $b_C = \sigma_C$ where $\sigma_C$ is a Rademacher random variable.

In our proof, we argue the following. Suppose that with high probability over the draw of the noisy $k$-XOR instance, there are degree-$O(\ell)$ Sum-of-Squares proofs that certify:
\begin{enumerate}[(1)]
\item (\cref{lem:localtoglobal}) For all $x \in \Fits^n$, $\abs{\E_{C \sim \cH_{2 \eps}}[\prod_{i \in C} x_i] - \frac{1}{n^k} (\sum_{i = 1}^n x_i)^k} \leq o(1)$ ;
\item (\cref{fact:hkmrefute}) For all $x \in \Fits^n$, $\abs{ \E_{C \sim \cH \setminus \cH_{2 \eps}}[\sigma_C \prod_{i \in C} x_i]} \leq o(\eps)$, where the $\sigma_C$'s are Rademacher random variables.
\end{enumerate}
Assuming that degree-$O(\ell)$ SoS can certify the above two inequalities, we can complete the proof. By a Chernoff bound, we have that $\abs{\cH_{2 \eps}} \approx 2 \eps m$, where $m = \abs{\cH}$.
Let $\pE_{\mu}$ denote the degree-$O(\ell)$ pseudoexpectation we find by maximizing the polynomial $\psi(x) =  \sum_{C \in \cH} b_C \prod_{i \in C} x_i$ subject to the constraints $x_i^2 = 1$ for all $i \in [n]$. We have $\psi(x^*) \geq 2 \eps (1- o(1))m$ (in expectation over the noisy $k$-XOR instance, $\psi(x^*) = 2 \eps m$, and then it is $2 \eps (1 - o(1)) m$ with high probability), and so we must have $\pE_{\mu}[\psi(x)] \geq 2 \eps(1 - o(1))m$. We can then split $\psi(x)$ as
\begin{equation*}
   \psi(x) = \underbrace{\sum_{C\in\cH_{2\eps}}b_C \prod_{i \in C} x_i}_{(\operatorname{I})} + \underbrace{\sum_{C\in\cH\setminus\cH_{2\eps}} b_C \prod_{i \in C} x_i}_{(\operatorname{II})} \mper
\end{equation*}
Let $\psi_1(x)$ denote the first term and $\psi_2(x)$ denote the second term. For each $C \in \cH \setminus \cH_{2 \eps}$, we have $b_C = \sigma_C$. Therefore, by Item (2), degree-$O(\ell)$ SoS can certify that $\psi_2(x) \leq o(\eps m)$ for all $x \in \Fits^n$, which implies that $\pE_{\mu}[\psi_2(x)] \leq o(\eps m)$. For the first term, we have that $b_C = \prod_{i \in C} x^*_i$ for all $C \in \cH_{2 \eps}$, and so by Item (1), it follows that $\abs{\pE_{\mu}[\frac{1}{2 \eps m}\psi_1(x)] - \pE_{\mu}[\frac{1}{n^k} (\sum_{i = 1}^n x_i x^*_i)^k]} \leq o(1)$.

Combining, we have shown that
\begin{flalign*}
\pE_{\mu}\left[\frac{1}{n^k} \ip{x,x^*}^k\right] &\geq \frac{1}{2 \eps m} \pE_{\mu}[\psi_1(x)] - o(1) = \frac{1}{2 \eps m} \pE_{\mu}[\psi(x) - \psi_2(x)] - o(1)  \\
&\geq  \frac{1}{2 \eps m} \left(2 \eps(1 - o(1))m - o(\eps m)  \right) - o(1) \geq 1 - o(1) \mper
\end{flalign*}
Hence, $\pE_{\mu}[\frac{1}{n^k} \ip{x,x^*}^k]$ is very close to $1$. From this, we would like to recover an assignment $\widehat{x}$ that is $o(1)$-close in relative Hamming distance to $x^*$. In the case of odd $k$, we do this by proving a generalization of~\cite[Lemma~A.5]{HopkinsSS15} (see \cref{fact:highcorrsos}), which shows that we can round by taking $\widehat{x}_i = \sgn(\pE_{\mu}[x_i])$. The case of even $k$ is slightly more complicated, as $\pE_{\mu}[x_i]$ may be $0$ for all $i$ because $x^*$ and $-x^*$ are equally valid planted assignments. We show that almost all of the rows of the matrix $\pE_{\mu}[x^{\otimes 2}] \in \R^{n \times n}$ must be highly correlated with $x^*$ (see \cref{lem:approxrecovery}), and from this we can extract a single $\widehat{x} \in \Fits^n$ that is $o(1)$-close to one of $x^*$, $-x^*$.

It remains to justify why degree-$O(\ell)$ SoS can certify Items (1) and (2) above. Item (2) is equivalent to requiring that the algorithm of~\cite{RaghavendraRS17,WeinAM19,GuruswamiKM22,HsiehKM23} for refuting random $k$-XOR can be ``captured'' by degree-$O(\ell)$ SoS, which is indeed true (see \cref{fact:hkmrefute}). Item (1) is a bit tricker to argue, but it turns out that we can reduce Item (1) to the case of Item (2) by using a symmetrization argument: we can replace $\frac{1}{n^k} (\sum_{i = 1}^n x_i)^k$ with $\E_{C \in \cH'}[ \prod_{i \in C} x_i]$, where $\cH'$ is a random hypergraph with the same distribution as $\cH_{2 \eps}$ and then view the random process of placing $C$ in $\cH_{2\eps}$ or $\cH'$ as assigning a Rademacher random variable $\sigma_C$ to each $C$. Item (2) is captured by \cref{lem:localtoglobal}, which we prove in \cref{sec:localtoglobal}.

As we can solve a degree-$O(\ell)$ SoS program in time $n^{O(\ell)}$, it follows that in $n^{O(\ell)}$ time we can approximately recover $x^*$ by finding an assignment $\widehat{x} \in \Fits^n$ that is $o(1)$-close to $x^*$.

\parhead{Rounding an approximate solution to an exact solution.} It remains to show how to round $\widehat{x}$ to exactly recover $x^*$. We will do this by using $\geq n\log n$ ``fresh'' constraints. That is, we divide the original $m$ constraints into two sets, and we use the first set to obtain $\widehat{x}$, and the second set to round $\widehat{x}$ to get $x^*$.

By ``shifting'' all equations by $\widehat{x}$, we can equivalently view this rounding task as recovering a planted assignment $y$ from a noisy $k$-XOR instance with $n \log n$  constraints under the assumption that $y$ has at most $o(n)$ entries that are $-1$. This is because $\widehat{x} \odot x^*$, where $x^*$ is the original assignment and $\odot$ denotes the element-wise product of vectors, has at most $o(n)$ entries that are $-1$.

Let us now look at all constraints that contain some fixed $i \in [n]$. By concentration, we should expect to have about $\log n$ of these constraints for each choice of $i$, and moreover each constraint, after removing the chosen index $i$, is distributed as a random set of size $k-1$. We can now make the following observation. Because $y$ has only $o(n)$ entries that are $-1$, a random set of size $k-1$ should typically avoid all $j$'s where $y_j = -1$. Therefore, a typical constraint $C \cup \{i\}$ where $\abs{C} = k-1$ will have $\prod_{j \in C} y_j = 1$, and so it will have $b_C = y_i \prod_{j \in C} y_j= y_i$ when the constraint is not noisy. Because all but a $o(1)$-fraction of such constraints are typical, by taking a majority vote we can recover each $y_i$ exactly with high probability, and therefore recover the planted assignment $y$.

The idea underpinning the above rounding procedure is simple: if $y$ is a very biased assignment, then a typical constraint contains only indices $j$ where $y_j = 1$, and so if we look at all constraints containing a certain index $i$, then the right-hand sides of these equations will be biased towards $y_i$, and this is simple to detect. We note that a similar idea was employed in the work of~\cite{BogdanovQ09}, which gave an algorithm to break Goldreich's PRG under the condition that the predicate has ``distribution complexity'' $r \leq 2$ (\cref{def:distcomp}).

\section{Preliminaries}
\label{sec:prelims}
For any $x\in\R^n$ and any $S\seq[n]$, define $x_S:= \prod_{i\in S}x_i$.
For any two vectors $x,y \in \R^n$, we define the correlation of $x, y$ to be $\corr(x, y):= \ip{x,y}/\norm{x}_2 \norm{y}_2$.

For any $z\in\R$, define 
\begin{equation*}
\sgn(z):= \begin{cases}
    1 & \text{if }z\geq 0\mcom\\
    -1 & \text{otherwise}\mper
\end{cases}
\end{equation*}
For any vector $x\in\R^n$, define $\sgn(x)\in \{\pm 1\}^n$ to be the entrywise application of $\sgn(\cdot)$ to $x$.

The following proposition shows that if $x \in \R^n$ is correlated with $x^* \in \Fits^n$, then $\sgn(x)$ is also correlated with $x^*$.
\begin{proposition}
\label{prop:realtoboolean}
    Let $x\in\R^n$ be a non-zero vector, and let $x^*\in\Fits^n$ be such that $\corr(x, x^*)\geq 1-\delta$. Then $\corr(\sgn(x), x^*)\geq 1 - 4\delta$.
\end{proposition}
\begin{proof}
Without loss of generality assume $x^* = \1$, and $x$ is a unit vector. Let $E:= \{i\in[n]:x_i < 0\}$. Since $\norm{x^*}_2 = \sqrt{n}$, we have that  
\begin{equation*}\sqrt{n}(1-\delta)\leq\left\langle x, x^*\right\rangle = \sum_{i: x_i\geq 0}x_i + \sum_{i: x_i < 0}x_i\leq\sum_{i: x_i\geq 0}x_i\implies\sum_{i: x_i\geq 0}x_i\geq(1-\delta)\sqrt{n}\mper\end{equation*}
At the same time, by Cauchy-Schwarz inequality, $\left(\sum_{i:x_i\geq 0}x_i\right)^2\leq(n-|E|)\cdot\sum_{i:x_i\geq 0}x^2_i\leq n - |E|$. Thus $n - |E|\geq (1 - \delta)^2n\implies |E|\leq (1 - (1 - \delta)^2)n\leq 2\delta n$. Consequently, $\corr(\sgn(x), x^*) = \frac{1}{n}\left(\sum_{i: x_i\geq 0}1-\sum_{i:x_i < 0}1\right) = 1-\frac{2|E|}{n}\geq 1 - 4\delta$, as desired.
\end{proof}

We define a $k$-uniform hypergraph to be a collection of \emph{tuples}  of size $k$. We also allow our hypergraphs to have repeated hyperedges, i.e., they can be multisets.

At times, it will be convenient to have notation to denote ordered tuples of length $k$ where the entries in the tuple are distinct. Towards this end, for any set $V$, we define $V^{(k)}:= \{x\in V^k: \text{All entries of }x\text{ are distinct}\}$. That is, $[n]^k$ is the set of tuples of length $k$, and $[n]^{(k)}$ is the set of tuples of length $k$ where all entries are distinct.

We state the following result to pass from tuples in $[n]^k$ to tuples with distinct entries:
\begin{proposition}
\label{prop:hypergraphclean}
    Let $k\geq 2$. Then $|[n]^k\setminus[n]^{(k)}|/n^k\leq k^2/n$. In particular, if $\cH$ is a collection of $m$ u.a.r.\ samples from $[n]^k$, then $|\cH\setminus[n]^{(k)}|/|\cH|\leq 2k^2/n$ with probability $\geq 1 - \exp(-\Omega(mk^2/n))$.
\end{proposition}
\begin{proof}
    Note that $|[n]^{(k)}| = n(n - 1)\cdots(n - k + 1) \geq n^k\left(1 - \frac{k - 1}{n}\right)^k\geq n^k\left(1 - \frac{k^2}{n}\right)$. Consequently, $\E|\cH\setminus[n]^{(k)}|\leq mk^2/n$, and by Chernoff, with probability $\geq 1 - \exp(-\Omega(mk^2/n))$, we have that $|\cH\setminus[n]^{(k)}|\leq 2mk^2/n$, as desired.
\end{proof}
We generate a random hypergraph with $m$ hyperedges on the vertex set $V$ by sampling $m$ hyperedges from $V^k$ uniformly at random.

We will also need the following concentration inequality.

\begin{fact}[McDiarmid's Inequality, \cite{McDiarmid89}]
 \label{fact:mcdiarmid}
    Let $f:\Omega^m\to\R$ be a function such that for any $i\in[m]$, and any $\omega_1, \ldots, \omega_m\in \Omega, \omega'_i\in \Omega$, we have that $|f(\omega_1, \ldots, \omega_i, \ldots, \omega_m) - f(\omega_1, \ldots, \omega'_i, \ldots, \omega_m)|\leq c_i$ for some real numbers $c_1, \ldots, c_m\geq 0$. Then if $X_1, \ldots, X_m\sim\Omega$ are independent, then 
    \begin{equation*}\Pr(|f(X_1, \ldots, X_m) - \E[f(X_1, \ldots, X_m)]|\geq\eps)\leq 2\exp\left(-\frac{2\eps^2}{\sum_{i\in[m]}c_i^2}\right)\mper\end{equation*}
\end{fact}

\subsection{Background on the Sum-of-Squares hierarchy}
\label{sec:sos}

We recall some basic facts about SoS (see \cite{FlemingKP19, BarakS16} for further details). Define $\R[x_1, \ldots, x_n]_{\leq t}$ to be the set of polynomials in $\R[x_1, \ldots, x_n]$ of degree $\leq t$.
\begin{definition}[Pseudo-expectations over the hypercube]
\label{sosaxioms}
    For any $d\geq 2$, a degree $d$ pseudo-expectation $\pE$ over $\Fits^n$ is a linear functional $\pE: \R[x_1, \ldots, x_n]_{\leq d}\to\R$ satisfying the following properties:
    \begin{enumerate}
        \item (Normalization) $\pE[1] = 1$, 
        \item (Booleanity) $\pE[fx_i^2] = \pE[f]$ for all $i\in[n], f\in\R[x_1, \ldots, x_n]_{\leq d - 2}$, 
        \item (Positivity) $\pE[f^2]\geq 0$ for all $f\in\R[x_1, \ldots, x_n]_{\leq d/2}$.
    \end{enumerate}
    Finally, denote by $\SoS_d$ $(\SoS_{\geq d})$ the set of all degree $d$ ($\geq d$) pseudo-expectations over $\Fits^n$.
\end{definition}
We will make a slight abuse of terminology and use the phrase ``pseudo-expectation'' to mean a pseudo-expectation over $\Fits^n$.

Suppose we have a function $f:\Fits^n\to\R$, and suppose the function attains its maximum on $A\seq\Fits^n$. It is easy to verify that the expectation operator associated to any probability distribution supported on $A$ is also a degree-$d$ pseudo-expectation for all $2\leq d\leq n$. In particular, $\max_{x\in\Fits^n}f(x)\leq \sup_{\pE\in\SoS_d}\pE[f]$ for all $2\leq d\leq n$. We also note the following elementary facts about pseudo-distributions on the hypercube:
\begin{fact}
\label{sosbasicfacts}
    Let $\pE$ be a pseudo-expectation of degree $\geq 2$ on the hypercube. Then:
    \begin{enumerate}[(1)]
        \item For all $i\in[n]$, $\pE[x_i]^2\leq 1$. In general, if $\pE$ is a pseudo-expectation of degree $\geq d$, and if $f\in\R[x_1,\ldots, x_n]_{\leq d/2}$, then $\pE[f]^2\leq\pE[f^2]$. In particular, for any $S\seq[n]$, we have $|\pE[x_S]|\leq 1$.
        \item   For all $f\in\R[x_1,\ldots, x_n]_{\leq d-2}$, we have $\pE\left[f\cdot \norm{x}_2^2 \right] = \sum_i \pE[f\cdot x_i^2] = \sum_i \pE[f]= n \pE[f]$.
    \end{enumerate}
\end{fact}
The SoS algorithm can compute all moments up to degree $\leq d$ of some pseudo-expectation in $\sup_{\pE\in\SoS_d}\pE[f]$ in $n^{O(d)}$ time:
\begin{fact}[SoS Algorithm (Corollary 3.40 in \cite{FlemingKP19})]
\label{sosfkp}
    Let $f = f(x_1, \ldots, x_n)$ be a polynomial of degree $t$ with rational coefficients such that each coefficient has $\poly(n)$ bit complexity. Then for any $d\geq t$, there exists an algorithm which on inputting $f$ and $d$ runs in time $n^{O(d)}$ and outputs $\{\alpha_S\}_{S\in\binom{[n]}{\leq d}}$, where $|\alpha_S - \pE_\mu[x_S]|\leq 4^{-n}$ for all $S\in\binom{[n]}{\leq d}$, and $\pE_\mu \in \arg\max_{\pE\in\SoS_d}\pE[f]$. In particular in $n^{O(d)}$ time one can compute $\alpha:= \pE_\mu[f]$ such that $\alpha$ satisfies $\beta + 2^{-n}\geq\alpha\geq\beta$, where $\beta:= \max_{x\in\Fits^n}f(x)$. 
\end{fact}

We also state a generalization of \cite[Lemma~A.5]{HopkinsSS15} which we need for ``rounding/decoding'' an approximately correct solution.
\begin{restatable}{fact}{highcorrsos}
\label{fact:highcorrsos}
    Let $k\geq 2$ be any integer. Let $\pE$ be a degree $(k + 2)$ pseudo-expectation on the hypercube such that $\pE[\langle x, x^*\rangle^k]\geq n^k(1 - \delta)$, for some $x^*\in\Fits^n, \delta\in(0, 1)$. Then:
    \begin{enumerate}[(1)]
        \item If $k$ is odd, then $\pE[\langle x, x^*\rangle]\geq n(1 - 2\delta)$, 
        \item If $k$ is even, then $\pE\left[\langle x, x^*\rangle^2\right]\geq n^2(1 - 2\delta)$.
    \end{enumerate}

\end{restatable}
We prove this statement in \cref{highcorrsosproof}.

\subsection{Sum-of-Squares refutation of \texorpdfstring{$k$}{k}-XOR}

We now define the canonical degree $d$ SoS relaxation associated to a system of equations:
For any hypergraph $\cH$, and any equation system $\{x_C = b_C\}_{C\in\cH}$, write $\psi(x):= \E_{C\sim\cH}[b_Cx_C] = \frac{1}{|\cH|}\sum_{C\in\cH}b_Cx_C$. The canonical degree $d$ SoS relaxation associated to $\psi$ (and hence to the system of equations) is $\arg\max_{\pE\in\SoS_d}\pE[\psi(x)]$. By \cref{sosfkp}, we can calculate any moment (up to degree $d$) of this relaxation in $n^{O(d)}$ time.

We crucially need the following theorem for refuting $k$-XOR systems due to Hsieh, Kothari, and Mohanty (\cite[Theorem~4.1]{HsiehKM23}):
\begin{fact}[Refuting $k$-XOR Systems]
\label{fact:hkmrefute}
    Let $k\in\N$, and let $\ell$ be a parameter such that $2k\leq\ell\leq n/8$. Let $\cH\seq[n]^{(k)}$ be any hypergraph such that 
    \begin{equation*}|\cH|\geq m_0:= 2^{O(k)}\cdot \left(\frac{n}{\ell}\right)^{k/2-1}\cdot\frac{n\log n}{\delta^4}\mcom\end{equation*} where $\delta\in(0, 1/2)$ is an arbitrary parameter. Consider the polynomial 
    \begin{equation*}\psi(x):= \E_{C\sim\cH}[\sigma_Cx_C]\mcom\end{equation*}
    where $\{\sigma_C\}_{C\in\cH}$ are i.i.d. uniform $\{-1, 1\}$-valued random variables. Then with probability $\geq 1 - 1/\poly(n)$ over the draw of $\{\sigma_C\}_{C\in\cH}$, for every pseudo-expectation $\pE$ of degree $\geq 2\ell$, we have $|\pE[\psi(x)]|\leq\delta$.
\end{fact}
\begin{remark}
    A few remarks are due:
    \begin{enumerate}[(1)]
        \item \cite{HsiehKM23} treats hypergraphs as collections of subsets of $[n]$. But note that any $\cH\seq[n]^{(k)}$ can be treated as a $k$-uniform hypergraph in the sense of \cite{HsiehKM23} by simply flattening tuples into sets. 
        \item The formal statement in~\cite[Theorem~4.1]{HsiehKM23} does not state that their algorithm is ``captured'' by degree-$2\ell$ Sum-of-Squares in the above sense. However, their algorithm goes via showing an upper bound on the spectral norm of a certain level $\ell$ Kikuchi matrix (\cref{def:evenkikuchi}), and it is straightforward to see that this reasoning is captured by a degree $2\ell$ Sum-of-Squares proof.
    \end{enumerate}
\end{remark}

In this paper, it will be more convenient to work in the case where $\cH$ is a subsample of $[n]^k$, rather than $[n]^{(k)}$. Because of this, we prove the following corollary using \cref{prop:hypergraphclean}, which is essentially just \cref{fact:hkmrefute} with this other distribution.
\begin{restatable}{corollary}{refuterandomhypergraph}
\label{cor:refuterandomhypergraph}
    Let $\cH$ be a collection of $m$ u.a.r.\ samples from $\cK:= [n]^k$. Let $\{\sigma_C\}_{C\in\cH}$ be i.i.d Rademacher random variables. If $2k\leq\ell\leq n/8$ is such that 
    \begin{equation*}
        m\geq \frac{2^{O(k)}}{\delta^4}\cdot\left(\frac{n}{\ell}\right)^{k/2-1}\cdot n\log n
    \end{equation*}
    for some $\delta\in(0, 1/2)$, then for any pseudo-expectation $\pE$ of degree $\geq 2\ell$, we have, with probability $\geq 1 - 1/\poly(n)$ over the randomness of $\cH$ and $\{\sigma_C\}_{C\in\cH}$,
    \begin{equation*}\left|\pE\left[\E_{C\sim\cH}\sigma_Cx_C\right]\right|\leq\delta + \frac{2k^2}{n}\mper\end{equation*}
\end{restatable}
We prove this statement in \cref{refuterandomhypergraphproof}.

\subsection{CSP to \texorpdfstring{$k$}{k}-XOR reduction}
We describe the CSP to XOR reduction we use to reduce solving random CSPs to solving random XOR systems, which has appeared in prior work~\cite{FeldmanPV15,GuruswamiHKM23}.

Fix a positive integer $m$, an assignment $x^*\in\Fits^n$, a predicate $P:\Fits^k\to\{0, 1\}$, and a planting distribution $\cQ$ for $P$. Sample $\Psi\sim\Psi(x^*, m, \cQ)$. For any $S\seq[k]$, define $\psi^{(S, +)}$ as the set of equations
\begin{equation*}\psi^{(S, +)}:= \left\{\prod_{j\in S}x_{i_j} = \prod_{j\in S}\literalneg(C)_j\right\}_{C = (i_1, \dots, i_k)\in\cH}\mcom\end{equation*}
and define $\psi^{(S, -)}$ similarly as
\begin{equation*}\psi^{(S, -)}:= \left\{\prod_{j\in S}x_{i_j} = -\prod_{j\in S}\literalneg(C)_j\right\}_{C = (i_1, \dots, i_k)\in\cH}\mper\end{equation*}
Finally, for any distribution $\cQ:\Fits^k\to[0, 1]$, and any $S\seq[k]$, define the Fourier coefficient 
\begin{equation*}\widehat{\cQ}(S):= \frac{1}{2^k}\sum_{y\in\Fits^k}\cQ(y) \prod_{j \in S} y_j\mper\end{equation*}
Note that $|\widehat{\cQ}(S)|\leq 2^{-k}$ for all $S\seq[k]$, since $\sum_{y\in\Fits^k}\cQ(y) = 1$. We recall the following standard fact from the analysis of Boolean functions:
\begin{fact}[Plancherel's Theorem]
 \label{plancherel}
    For any function $f:\Fits^k\to\R$, we have 
    \begin{equation*}\frac{1}{2^k}\sum_{y\in\Fits^k}f(y)^2 = \sum_{S\seq[k]}\widehat{f}(S)^2\mcom\end{equation*}
    where recall that $\widehat{f}(S):= 2^{-k}\sum_{y\in\Fits^k}f(y)\prod_{j\in S}y_j$.
\end{fact}
We now show that if $\cQ$ does not have full-support (which will be the case when $P$ is a non-trivial predicate, i.e.\ $P^{-1}(1)\neq\Fits^k$), then it has a non-trivially large non-zero Fourier coefficient: 
\begin{proposition}
\label{largestfouriercoefflowerbound}
    Let $\cQ$ be a probability distribution on $\Fits^k$ such that $\supp(\cQ)\neq\Fits^k$. Then $\max_{\emptyset \neq S\subseteq[k]  }|\widehat{\cQ}(S)| > 4^{-k}$.
\end{proposition}
\begin{proof}
Since $\cQ$ is a probability distribution, $\sum_{y\in\Fits^k}\cQ(y) = 1$. Then by Cauchy-Schwarz, 
\begin{equation*}\sum_{y\in\Fits^k}\cQ(y)^2\geq\frac{1}{|\supp(\cQ)|}\geq\frac{1}{2^k - 1}\mper\end{equation*}
Consequently, by \cref{plancherel}, 
\begin{equation*}\sum_{S\seq[k]}\widehat{\cQ}(S)^2\geq\frac{1}{2^k(2^k - 1)}\mper\end{equation*}
Now, $\widehat{\cQ}(\emptyset) = \frac{1}{2^k}\sum_{y\in\Fits^k}\cQ(y) = 2^{-k}$, and thus 
\begin{equation*}\sum_{\emptyset\neq S\seq[k]}\widehat{\cQ}(S)^2\geq\frac{1}{2^k(2^k - 1)} - \frac{1}{(2^k)^2} = \frac{1}{4^k(2^k - 1)}\mper\end{equation*}
Consequently, 
\begin{equation*}\max_{\emptyset\neq S\seq[k]}|\widehat{\cQ}(S)|^2\geq\frac{1}{2^k - 1}\cdot\frac{1}{4^k(2^k - 1)} > \frac{1}{16^k}\implies\max_{\emptyset\neq S\seq[k]}|\widehat{\cQ}(S)| > 4^{-k}\mper\qedhere\end{equation*}
\end{proof}

With the above notation, we will use the following reduction from CSPs to $k$-XOR from \cite{GuruswamiHKM23}:
\begin{fact}[Claim 4.2, \cite{GuruswamiHKM23}]
\label{csptoxorreduction}
    Let $\Psi\sim\Psi(x^*, m, \cQ)$ for some planting distribution $\cQ$ for a nontrivial predicate $P$.
    Then, for all non-empty $S\seq[k]$, $\psi^{(S, +)}$ is distributed as $\LPN_{|S|}\left(x^*, m, 2^{k-1}\widehat{\cQ}(S)\right)$ and $\psi^{(S, -)}$ is distributed as $\LPN_{|S|}\left(x^*, m, -2^{k-1}\widehat{\cQ}(S)\right)$.
 
Furthermore, if the hypergraph of $\Psi$ is $\cH$, then the underlying hypergraph of $\psi^{(S, \pm)}$ is $\cH\vert_S:= \{C\vert_S: C\in\cH\}$, where for any $C = (i_1, \ldots, i_k)\in[n]^k$ and $S \subseteq [k]$, we define $C\vert_S := (i_j :j\in S)$.
\end{fact}

\section{From Planted CSPs to Noisy \texorpdfstring{$k$}{k}-XOR}
\label{sec:plantedcsp}
In this section, we present our algorithm for solving random planted CSPs and prove its guarantees (\cref{mthm:csp}) using our algorithm for noisy $k$-XOR (\cref{mthm:noisykxor}) as a subroutine.

We begin by recalling the statement of \cref{mthm:csp}.
\cspalg*

We now prove \cref{mthm:csp}, using \cref{mthm:noisykxor}. 

\begin{proof}[Proof of \cref{mthm:csp} from \cref{mthm:noisykxor}]
We shall show that the following algorithm satisfies all the required properties.
\begin{mdframed}
    \begin{algorithm}[Solving Random Planted CSPs]
    \label{alg:solverandomcsp}    
    \mbox{}
    \begin{description}
        \item[Input:] A CSP predicate $P:\Fits^k\to\{0, 1\}$, a CSP system $\Psi = \{P(x_{i_1}, \ldots, x_{i_k}) = 1\}_{C = (i_1, \dots, i_k)\in\cH}$, where $\cH\seq[n]^k$, and a parameter $2k\leq\ell\leq n/8$.
        \item[Output:] A vector $x\in\Fits^n$.
        \item[Operation:] \mbox{}
        \begin{enumerate}[(1)]
            \item If $P(y) = 1$ for all $y\in\Fits^k$, then return $x:= \1$.
            \item Initialize $\mathcal{V}:= \emptyset$.
            \item \textbf{For} $\emptyset\neq S\seq[k]$, \textbf{For} $\ast\in\{+, -\}$:
            \begin{enumerate}[(a)]
                \item Run the algorithm of \cref{mthm:noisykxor} (\cref{alg:recxxor})  on input $(\psi^{(S, \ast)}, \ell)$ to obtain output $x\in\Fits^n$. Add $x, -x$ to $\mathcal{V}$.
            \end{enumerate}
            \item For every $v\in\mathcal{V}$, check if $v$ satisfies $\Psi$. Return the first satisfying assignment found.
            \item If no satisfying assignment is found in the previous line, return $\1$. 
        \end{enumerate}
    \end{description}
    \end{algorithm}
\end{mdframed}
    If $P(y) = 1$ for all $y\in\Fits^k$, then the CSP instance is trivially satisfiable, and the algorithm may return any vector, say $\1$.

    Thus assume $P^{-1}(1)\subsetneq \Fits^k$, and thus $\supp(\cQ)\neq\Fits^k$. Let $S^*\seq[k]$ be the smallest non-empty set such that $|\widehat{\cQ}(S^*)|\geq 4^{-k}$ (such a set exists due to \cref{largestfouriercoefflowerbound}). Write $r = |S^*|$, and note that $r$ is the distribution complexity of $\cQ$.

    Now, enumerate all non-empty $S\seq[k]$, and for each $S$ guess if $\widehat{\cQ}(S)$ is $ > 0, = 0$ or $< 0$. For each $S$ where we guessed $\widehat{\cQ}(S) > 0$ (resp.\ $< 0$), construct the instance $\psi:= \psi^{(S, +)}$ (resp.\ $\psi:= \psi^{(S, -)}$), run \cref{alg:recxxor} on $\psi$ and note down its output. 
    
    Consequently, we obtain a list of $4(2^k - 1) < 2^{k + 2}$ outputs from the above procedure for $2^k - 1$ different non-empty subsets of $[k]$, $2$ possible signs of each Fourier coefficient, and another factor of $2$ to include the negation of each output when $k$ is even. 
    
    Now, note that since $\cH$ is a collection of $m$ u.a.r.\ samples from $[n]^k$, $\cH\vert_{S^*}$ is a collection of $m$ u.a.r.\ samples from $[n]^{|S^*|} = [n]^r$. Consequently, when we guess the sign of $\widehat{\cQ}(S^*)$ correctly, running \cref{alg:recxxor} on $\psi^{(S^*, \sgn(\widehat{\cQ}(S^*)))}\sim \LPN_{r}\left(x^*, m, \eps\right)$, where $\eps\geq 2^{k - 1}\cdot 4^{-k}\geq 2^{-O(k)}$, yields a vector $x\in\{\pm x^*\}$ with probability $\geq 1 - 1/\poly(n)$, by \cref{csptoxorreduction} and \cref{mthm:noisykxor}. Consequently, with probability $\geq 1 - 1/\poly(n)$, $\mathcal{V}$, which is a list of size $\leq 2^{k + 2}$, contains $x^*$, which is a satisfying assignment for $\Psi$, as desired.
    
    Finally, the algorithm described above takes $O(2^k)\cdot n^{O(\ell)} = n^{O(\ell)}$ time to run, as $\ell \geq 2k$.
\end{proof}

\section{Solving Noisy \texorpdfstring{$k$}{k}-XOR}
\label{sec:noisykxor}
In this section, we prove \cref{mthm:noisykxor}, which is our algorithm to solve noisy $k$-XOR instances, or equivalently $k$-sparse LPN. Below, we recall the statement of \cref{mthm:noisykxor}

\noisykxor*

As explained in \cref{sec:proofoverview}, the algorithm of  \cref{mthm:noisykxor} proceeds in two steps. First, we will recover an approximate solution $\widehat{x}$ is that is close to $x^*$, and then we will round $\widehat{x}$ to recover $x^*$ (up to a global sign). The two steps are captured by the following two lemmas below, which we will prove in \cref{sec:approxrecovery,sec:rounding}, respectively.

\begin{lemma}[Finding an approximate solution]
   \label{lem:approxrecovery}
    Let $m,n$ be positive integers, and let $2\leq k\leq o(n^{1/4})$. Let $\ell$ be a parameter such that $2k\leq \ell\leq n/8$. There is an algorithm that takes as input a $k$-XOR instance $(\cH, \{b_C\}_{C \in \cH})$ and in time $n^{O(\ell)}$ outputs an assignment $\widehat{x} \in \Fits^n$
    with the following guarantee. Suppose that $(\cH, \{b_C\}_{C \in \cH}) \sim\LPN_k(x^*, m, \eps)$, where $\eps\in(0, \frac{1}{2}]$ and 
    \begin{align}
        m\geq\frac{2^{O(k)}}{\eps^6\delta^5}\cdot\left(\frac{n}{\ell}\right)^{k/2-1}\cdot n\log n
    \end{align}
    for some $\delta\in(0, 1/2)$, such that $\eps\delta\geq\Omega(n^{-1/2})$ for some large enough constant in the $\Omega(\cdot)$. Then with probability at least $1 - 1/\poly(n)$, it holds that
\begin{enumerate}[(1)]
    \item If $k$ is odd, then $\corr(\widehat{x}, x^*)\geq 1 - \delta$;
    \item If $k$ is even, then either $\corr(\widehat{x}, x^*)\geq 1 - \delta$ or $\corr(-\widehat{x}, x^*)\geq 1 - \delta$.
\end{enumerate}
\end{lemma}

\begin{lemma}[Rounding an approximate solution to an exact solution]\label{lem:rounding}
Let $m,n$ be positive integers and let $2\leq k\leq o(\sqrt{n})$. There is an algorithm that takes as input a $k$-XOR instance $(\cH, \{b_C\}_{C \in \cH})$ along with an assignment $\widehat{x} \in \Fits^n$, and in time $\poly(m,n)$ outputs an assignment $\widetilde{x} \in \Fits^n$ with the following guarantee. If $(\cH, \{b_C\}_{C \in \cH}) \sim \LPN_k(x^*, m, \eps)$ where $\corr(\widetilde{x}, x^*)\geq 1 - \delta$ (if $k$ is odd) or $\abs{\corr(\widetilde{x}, x^*)} \geq 1 - \delta$ (if $k$ is even) for some $\delta$ such that $\delta\leq\frac{\eps}{k}$, and if $m\geq\Omega(\eps^{-2}n\log n)$, then with probability $\geq 1 - 1/\poly(n)$ over the draw of the noisy $k$-XOR instance, the output $\widetilde{x}$ satisfies $\widetilde{x} = x^*$ (if $k$ is odd) or $\widetilde{x} \in \{x^*, -x^*\}$ (if $k$ is even).
\end{lemma}
We now prove \cref{mthm:noisykxor} from \cref{lem:approxrecovery,lem:rounding}.

\begin{proof}[Proof of \cref{mthm:noisykxor} from \cref{lem:approxrecovery,lem:rounding}]
\label{alg:recxxorproof}

We shall show that the following algorithm satisfies all the required properties:

\begin{mdframed}
    \begin{algorithm}[Recovering $x^*$ from Noisy XOR system]
    \label{alg:recxxor}    
    \mbox{}
    \begin{description}
        \item[Input:] A system of $k$-XOR equations $(\cH, \{b_C\}_{C \in \cH})$, where $\cH\seq[n]^k$, and a parameter $2k\leq\ell\leq n/8$.
        \item[Output:] A vector $x\in\Fits^n$.
        \item[Operation:] \mbox{}
        \begin{enumerate}[(1)]
            \item \textbf{If} $k \geq 2$: 
            \begin{enumerate}[(a)]
                \item Split $\cH$ into two parts $\cH = \cH_1\sqcup\cH_2$ with $|\cH_1| = \lceil |\cH|/2\rceil, |\cH_2| = \lfloor |\cH|/2\rfloor$.
                \item Run the algorithm of \cref{lem:approxrecovery} (\cref{alg:approxrecovery}) on input $(\cH_1, \{b_C\}_{C \in \cH_1})$, and let $\widehat{x}$ denote the output.
                \item Run the algorithm of \cref{lem:rounding} (\cref{alg:rounding}) on input $(\widehat{x},\cH_2, \{b_C\}_{C \in \cH_2})$, and let $\widetilde{x}$ denote the output. Return $\widetilde{x}$.
            \end{enumerate}
            \item \textbf{Else if} $k = 1$: 
            \begin{enumerate}[(a)]
                \item For every $i\in[n]$, collect all equations of the form $\{x_i = b_{i_j}\}$, and return $x\in\Fits^n$, where $x_i:= \operatorname{majority}(\{b_{i_j}\})$.
            \end{enumerate}
        \end{enumerate}
    \end{description}
    \end{algorithm}
\end{mdframed}

    First assume $k\geq 2$, and write $\delta:= \eps/k$. Split $\cH$ into two parts $\cH = \cH_1\sqcup\cH_2$ with $|\cH_1| = \lceil m/2\rceil, |\cH_2| = \lfloor m/2\rfloor$. Note that $|\cH_2|\geq\Omega(\eps^{-2}n\log n)$. We now make cases:
    \begin{enumerate}[(1)]
        \item \textbf{$k$ odd}: In this case, applying \cref{lem:approxrecovery} in Step (1b), with probability $1 - 1/\poly(n)$ we obtain $\widehat{x}\in \Fits^n$ with $\corr(\widehat{x}, x^*)\geq 1 - \delta = 1 - \eps/k$. Now, note that the distribution of $\{\cH_1, \{b_C\}_{C \in \cH_1}\}$ (which includes the randomness of $\cH_1$ and the randomness of $\{b_C\}$) is independent of $\{\cH_2, \{b_C\}_{C \in \cH_2}\}$. Since $\widetilde{x}$ is completely determined by $\{\cH_1, \{b_C\}_{C \in \cH_1}\}$, $\widetilde{x}$ is independent of $\{\cH_2, \{b_C\}_{C \in \cH_2}\}$. Consequently, we can apply \cref{lem:rounding} to recover $x^*$ with probability $\geq 1 - 1/\poly(n)$, as desired.
        \item \textbf{$k$ even}: In this case, applying \cref{lem:approxrecovery} in Step (1b), with probability $1 - 1/\poly(n)$ we obtain $\widehat{x} \in \Fits^n$  such that $\abs{\corr(\widehat{x}, x^*)}\geq 1 - \delta$. Again, by the same argument as in Item (1), we have that $\widetilde{x}$ is independent of $\{\cH_2, \{b_C\}_{C \in \cH_2}\}$. Consequently, we can apply \cref{lem:rounding} to recover $x^*$ up to a sign with probability $\geq 1 - 1/\poly(n)$, as desired.
\end{enumerate}

Observe that for $k\geq 2$, the runtime of the algorithm is $n^{O(\ell)} + \poly(m,n)\leq n^{O(\ell)}$, using the runtime bounds in \cref{lem:approxrecovery,lem:rounding}.

Finally, suppose $k = 1$. Note that for any $i\in[n]$, the expected multiplicity of $(i)$ in $\cH$ is $|\cH|/n$. Consequently, by Chernoff $+$ union bound, with probability $\geq 1 - 1/\poly(n)$, for every $i\in[n]$, we have $\geq\Omega(\eps^{-2}\log n)$ equations of the form $\{x_i = b_{i_j}\}$. Furthermore, $b_{i_j} = x^*_i$ with probability $\geq 1/2 + \eps$, and consequently, $\operatorname{majority}(\{b_{i_j}\}) = x^*_i$ with probability $\geq 1 - \exp(-\Omega(\eps^2\cdot (\eps^{-2}\log n)))\geq 1 - 1/\poly(n)$, and the theorem follows by a union bound over all $i\in[n]$.
\end{proof}
\begin{remark}
A few remarks are due:
\begin{enumerate}[(1)]
    \item Note that \cref{alg:recxxor} doesn't assume knowledge of $\eps$. \cref{alg:recxxor} only assumes knowledge of $\ell$, which should be viewed as a proxy for the amount of runtime we have available. Given $\ell$, if $\eps$ is large enough (as dictated by \cref{eq:mlowerboundmaintheorem}), then we recover $x^*$ successfully.
    \item By \cite[Remark~2]{HsiehKM23}, for even $k$, we can actually take $m_0:= 2^{O(k)}\cdot \left(\frac{n}{\ell}\right)^{k/2-1}\cdot\frac{n\log n}{\delta^2}$, as opposed to the $\delta^{-4}$ dependence right now in \cref{fact:hkmrefute}. Tracking the dependence of $\delta$ throughout the proof of \cref{mthm:noisykxor}, for even $k$, we can take 
    \begin{equation*}m_0:= \frac{2^{O(k)}}{\eps^7}\cdot \left(\frac{n}{\ell}\right)^{k/2-1}\cdot n\log n\mcom\end{equation*}
    in \cref{eq:mlowerboundmaintheorem}, instead of the $\eps^{-11}$ dependence.
\end{enumerate}
\end{remark}

\subsection{Finding an approximate solution: proof of \texorpdfstring{\cref{lem:approxrecovery}}{Lemma~\ref{lem:approxrecovery}}}
\label{sec:approxrecovery}

In this subsection, we prove \cref{lem:approxrecovery}. We do this by showing that the algorithm below has the required guarantees.

\begin{mdframed}
    \begin{algorithm}[Recovering $x^*$ from Noisy XOR system]
    \label{alg:approxrecovery}    
    \mbox{}
    \begin{description}
        \item[Input:] A system of equations $(\cH, \{b_C\}_{C\in\cH})$ and a parameter $2k\leq\ell\leq n/8$.
        \item[Output:] A vector $x\in\Fits^n$.
        \item[Operation:] \mbox{}
            \begin{enumerate}[(1)]
                \item \textbf{If} $k$ is odd:
                \begin{enumerate}[(i)]
                    \item\label{oddkprocedure} Compute $x' := \pE_\mu[x]$, where $\pE_\mu\in\arg\max_{\pE\in\SoS_{2\ell}}\pE[\sum_{C\in\cH}b_Cx_C]$. 
                    \item Output $\widehat{x}:= \sgn(x')$.
                \end{enumerate}
                \item \textbf{Else if} $k$ is even:
                \begin{enumerate}[(i)]
                    \item Compute $X:= \pE_\mu[x^{\otimes 2}]\in\R^{n\times n}$, where $\pE_\mu\in\arg\max_{\pE\in\SoS_{2\ell}}\pE[\sum_{C\in\cH}b_Cx_C]$. 
                    \item\label{evenkprocedure} For each row $X^{(i)}$ of $X$, compute $\widehat{x}^{(i)}:= \sgn(X^{(i)})$.
                    \item\label{item:nearestdecode} For each $i,j \in [n]$, compute $\abs{\corr(\widehat{x}^{(i)}, \widehat{x}^{(j)})}$. For each $i \in [n]$, let $\delta_i$ be such that $\abs{\corr(\widehat{x}^{(i)}, \widehat{x}^{(j)})} \geq 1 - \delta_i$ for at least $0.99 n$ choices of $j \in [n]$. Let $i^* \in [n]$ be a minimizer of $\delta_i$. Output $\widehat{x}^{(i^*)}$.
                \end{enumerate}
                
        \end{enumerate}
    \end{description}
    \end{algorithm}
\end{mdframed}

We start by showing that the canonical degree $2\ell$ SoS program recovers a vector which is strongly correlated with $x^*$.

Observe that in Step (1i) or Step (2i), \Cref{alg:approxrecovery} first finds a pseudoexpectation which maximizes
$\pE[\sum_{C\in\cH}b_Cx_C]$. We show that for such a pseudodistribution, $\pE[\sum_{C\in\cH}b_Cx_C]$ is close to $\pE[\ip{x,x^*}^k]$ (when we normalize appropriately) with high probability over the choice of a random XOR instance, which we use to bound the correlation.
\begin{lemma}
\label{pseudodistrefute}
    Let $2\leq k\leq o(n^{1/4})$, and let $(\cH, \{b_C\}_{C\in\cH})\sim\LPN_k(x^*, m, \eps)$, where $\eps\in(0, \frac{1}{2}]$. Let $\ell$ be a parameter such that $2k\leq \ell\leq n/8$, and suppose 
    \begin{align}
    \label{mlowerbound}
        m\geq\frac{2^{O(k)}}{\eps\delta^5}\cdot\left(\frac{n}{\ell}\right)^{k/2-1}\cdot n\log n\mcom
    \end{align}
    where $\delta\in(0, 1/2)$ is a parameter. Then with probability $\geq 1 - 1/\poly(n)$ over the draw of $\cH$ and $\{x_C = b_C\}_{C\in\cH}$, we have that for any pseudo-expectation $\pE$ of degree $\geq 2\ell$, 
    \begin{equation*}\pE\left[\psi(x) - \frac{2\eps}{n^k}\langle x, x^*\rangle^k\right]\leq\delta + O\left(\frac{1}{\sqrt{n}}\right) \mcom\end{equation*}
    where $\psi(x) := \E_{C\sim\cH}[b_Cx_C]$.
\end{lemma}
\begin{proof}
Consider the random variable 
\begin{equation*}X:= \begin{cases}
    \sigma & \text{with probability }1 - 2\eps\mcom\\
    1 & \text{with probability }2\eps
\end{cases}\mcom\end{equation*}
where $\sigma$ is a Rademacher random variable, i.e.\ it takes values $\pm 1$ with equal probability. Note that $\Pr(X = 1) = 1 - \Pr(X = -1) = \frac{1}{2} + \eps$, and thus $b_C$ is distributed as $X\cdot x^*_C$. Consequently, we can write 
\begin{align}
\label{mainterm}
    \pE\left[\sum_{C\in\cH}b_Cx_C - \frac{2\eps m}{n^k}\langle x, x^*\rangle^k\right] = \underbrace{\pE\left[\sum_{C\in\cH_{2\eps}}(x\odot x^*)_C - \frac{2\eps m}{n^k}\langle x, x^*\rangle^k\right]}_{(\operatorname{I})} + \underbrace{\pE\left[\sum_{C\in\cH\setminus\cH_{2\eps}}\sigma_C(x\odot x^*)_C\right]}_{(\operatorname{II})}\mcom
\end{align}
where $\cH_{2\eps}$ is a $2\eps$-sample of $\cH$, i.e.\ every element of $\cH$ is included in $\cH_{2\eps}$ independently with probability $2\eps$. We now bound the above terms one by one:

\parhead{Bounding term $(\operatorname{I})$.} Write $m':= |\cH_{2\eps}|$. Note that $\E[m'] = 2\eps m$, and by a Chernoff bound, 
\begin{align}
\label{mconc}
    m'\in 2\eps m\cdot \left[1\pm O\left(\frac{1}{\sqrt{n}}\right)\right]
\end{align}
with probability $\geq 1 - \exp(-\Omega(\eps m/n))\geq 1 - \exp(-\Omega(\log n))\geq 1 - 1/\poly(n)$, since $\eps m\geq 2^{O(k)}\frac{n^{\frac{k}{2}}\log n}{\ell^{\frac{k}{2} - 1}}\geq 2^{O(k)}n\log n$, since $\ell\leq n$.
Equivalently, $m'\geq\frac{2^{O(k)}}{\delta^5}\cdot\frac{n^{\frac{k}{2}}\log n}{\ell^{\frac{k}{2} - 1}}$ with probability $\geq 1 - 1/\poly(n)$.

To finish bounding this term, we will need the following lemma, which intuitively shows that degree-$2\ell$ Sum-of-Squares can certify that a random hypergraph with a sufficient number of samples is a good approximation of the complete hypergraph.
\begin{restatable}{lemma}{localtoglobal}
\label{lem:localtoglobal}
    Let $\cH$ be a collection of $m$ u.a.r.\ samples from $\cK:= [n]^k$, where $2\leq k\leq o(n^{1/4})$. Let $\delta\in(0, 1/2)$ be an arbitrary parameter. Then for any pseudo-expectation $\pE$ of degree $\geq 2\ell$, where 
    \begin{equation*}
        \frac{n}{8}\geq \ell\geq 2k \mcom \, \ell^{\frac{k}{2} - 1}\geq 2^{O(k)}\cdot \frac{n^{\frac{k}{2}}}{m}\cdot\frac{\log n}{\delta^4} \mcom
    \end{equation*}
    we have, with probability $\geq 1 - 1/\poly(n)$ over the draw of $\cH$,
    \begin{equation*}
        \pE\left[\E_{C\sim\cH}[x_C] - \E_{C\sim\cK}[x_C]\right]\leq\delta + O\left(\frac{1}{\sqrt{n}}\right)\mper
    \end{equation*}
\end{restatable}
We postpone the proof of \cref{lem:localtoglobal} to \cref{sec:localtoglobal}, and now use it to finish bounding term (I).

Since $\cH_{2\eps}$ is distributed as a collection of $m'$ u.a.r.\ samples of $\cK$, by \cref{lem:localtoglobal}, for any $\pE\in\SoS_{\geq 2\ell}$, we have that 
\begin{equation*}\pE\left[\E_{C\sim\cH_{2\eps}}[y_C] - \E_{C\sim\cK}[y_C]\right]\leq\delta + O\left(\frac{1}{\sqrt{n}}\right)\implies\pE\left[\sum_{C\in\cH_{2\eps}}y_C - m'\E_{C\sim\cK}y_C\right]\leq\left(\delta + O\left(\frac{1}{\sqrt{n}}\right)\right)m'\mcom\end{equation*}
where $y:= x\odot x^*$.\footnote{Note that technically \cref{lem:localtoglobal} is stated in terms of $x$, but since it works for all pseudo-distributions of degree $\geq 2\ell$, by translating we can put $y$ instead of $x$.} Now note that 
\begin{align}
\label{dotprodrel}
    \E_{C\sim\cK}[y_C] = \E_{C\sim\cK}[(x\odot x^*)_C] = \frac{1}{n^k}\sum_{C\in[n]^k}(x\odot x^*)_C = \frac{1}{n^k}\langle x, x^*\rangle^k\mper
\end{align}
Consequently, 
\begin{align*}
    \left(\delta + O\left(\frac{1}{\sqrt{n}}\right)\right)m'
    &\geq \pE\left[\sum_{C\in\cH_{2\eps}}y_C - m'\E_{C\sim\cK}y_C\right] \\
    &= \pE\left[\sum_{C\in\cH_{2\eps}}y_C - \frac{2\eps m}{n^k}\langle x, x^*\rangle^k\right] + \left(\frac{2\eps m - m'}{n^k}\right)\pE\langle x, x^*\rangle^k
\end{align*}
\begin{align}
\label{term1bound}
    \implies \pE\left[\sum_{C\in\cH_{2\eps}}y_C - \frac{2\eps m}{n^k}\langle x, x^*\rangle^k\right]\leq\left(\delta + O\left(\frac{1}{\sqrt{n}}\right)\right)m' +  \left|m'-2\eps m\right|\cdot\left|\frac{\pE\langle x, x^*\rangle^k}{n^k}\right|\mper
\end{align}
Now, by \cref{dotprodrel}, we have $\frac{\pE[\langle x, x^*\rangle^k]}{n^k} = \pE[\E_{C\sim\cK}[y_C]] = \E_{C\sim\cK}[\pE[x_C]]$. By \cref{sosbasicfacts}, $\pE[x_C]\in[-1, 1]$ for all $C\in\cK$, and thus $\E_{C\sim\cK}\pE[x_C]\in[-1, 1]\implies\left|\frac{\pE[\langle x, x^*\rangle^k]}{n^k}\right|\leq 1$. Also, by \cref{mconc}, $|m' - 2\eps m|\leq O(m'/\sqrt{n})$. Putting all this together in \cref{term1bound}, we obtain that 
\begin{equation*}\pE\left[\sum_{C\in\cH_{2\eps}}y_C - \frac{2\eps m}{n^k}\langle x, x^*\rangle^k\right]\leq m'\cdot\left(\delta + O\left(\frac{1}{\sqrt{n}}\right) +  O\left(\frac{1}{\sqrt{n}}\right)\right) \leq m\left(\delta + O\left(\frac{1}{\sqrt{n}}\right)\right)\mper\end{equation*}

\parhead{Bounding term $(\operatorname{II})$.} Write $m'':= |\cH\setminus\cH_{2\eps}|$. 
Here we make two cases:
\begin{enumerate}[(1)]
    \item $m''\leq\delta m$: In this case, $\pE\left[\sum_{C\in\cH\setminus\cH_{2\eps}}\sigma_C(x\odot x^*)_C\right] = \sum_{C\in\cH\setminus\cH_{2\eps}}\sigma_C\pE\left[(x\odot x^*)_C\right]\leq m''\leq \delta m$.
    \item $m'' > \delta m$: Note that $m''> \delta m\geq \frac{2^{O(k)}}{\eps\delta^4}\cdot\frac{n^{\frac{k}{2}}\log n}{\ell^{\frac{k}{2} - 1}}\geq \frac{2^{O(k)}}{\delta^4}\cdot\frac{n^{\frac{k}{2}}\log n}{\ell^{\frac{k}{2} - 1}}$. Also note that $\cH\setminus\cH_{2\eps}$ is a collection of $m''$ u.a.r.\ samples from $\cK:= [n]^k$. Consequently, by \cref{cor:refuterandomhypergraph}, for any $\pE\in\SoS_{\geq 2\ell}$, we have $\pE\sum_{C\in\cH\setminus\cH_{2\eps}}\sigma_Cx_C\leq m''\left(\delta + \frac{2k^2}{n}\right)\leq m\left(\delta + \frac{2k^2}{n}\right)$.
\end{enumerate}
Combining the above terms together, we can say that for any $\pE\in\SoS_{\geq 2\ell}$, we have 
\[\pE\left[\sum_{C\in\cH\setminus\cH_{2\eps}}\sigma_Cx_C\right]\leq m\left(\delta + \frac{2k^2}{n}\right)\leq m\left(\delta + O(n^{-1/2})\right)\mcom\]

where the last inequality follows since $k\leq o(n^{1/4})$. Finally, putting everything together in \cref{mainterm}, we obtain that 
\begin{equation*}\pE\left[\sum_{C\in\cH}b_Cx_C - \frac{2\eps m}{n^k}\langle x, x^*\rangle^k\right] \leq m\left(\delta + O\left(\frac{1}{\sqrt{n}}\right)\right) + m\left(\delta + O\left(\frac{1}{\sqrt{n}}\right)\right)\end{equation*}
\begin{equation*}\leq m\left(2\delta + O\left(\frac{1}{\sqrt{n}}\right)\right)\mcom\end{equation*}
and we absorb the $2$ in front of the $\delta$ in the $2^{O(k)}$ term in \cref{mlowerbound}.
\end{proof}
We can now prove the main lemma.
\begin{proof}[Proof of \Cref{lem:approxrecovery}]
Write $\nu:= \eps\delta$. Note that $m\geq\frac{2^{O(k)}}{\eps\nu^5}\cdot\left(\frac{n}{\ell}\right)^{k/2-1}\cdot n\log n$.

Let $\pE_\mu\in\arg\max_{\pE\in\SoS_{2\ell}}\psi(x)$ be any pseudo-distribution maximizing $\psi(x):= \E_{C\sim\cH}[b_Cx_C]$. By \cref{pseudodistrefute}, $\pE_\mu\left[\psi(x) - \frac{2\eps}{n^k}\langle x, x^*\rangle^k\right]\leq\nu + O\left(\frac{1}{\sqrt{n}}\right)\leq 2\nu\implies \pE_\mu[\psi(x)] \leq \frac{2\eps}{n^k}\pE_\mu\langle x, x^*\rangle^k + 2\nu$, where we can write $\nu + O(n^{-1/2})\leq 2\nu$ since $\nu = \eps\delta\geq\Omega(n^{-1/2})$.

Now, consider the probability distribution which places all its mass on $x^*$. By \cref{sosaxioms}, the expectation operator associated to this distribution is a valid degree $2\ell$ pseudo-distribution over $\Fits^n$. Furthermore, the expectation of $\psi(x)$ under this distribution is $\psi(x^*)$. Consequently, $\pE_\mu\psi(x)\geq\psi(x^*)$, and thus 
    \begin{align}
    \label{dotproductlowerbound}
    \frac{2\eps}{n^k}\pE_\mu\langle x, x^*\rangle^k + 2\nu 
    &\geq \psi(x^*) = \E_{C\sim\cH}[b_C x^*_C] 
    \overset{(*)}{\geq} 2\eps\left(1 - O\left(\frac{1}{\sqrt{n}}\right)\right) \\
    \implies \frac{1}{n^k} \pE_\mu\langle x, x^*\rangle^k 
    &\geq 1 - \frac{\nu}{\eps} - O\left(\frac{1}{\sqrt{n}}\right) =: 1 - \xi \mcom
\end{align}
    where $(\ast)$ follows with probability $\geq 1 - 1/\poly(n)$ by a Chernoff bound. Note that $\xi = \nu/\eps + O(n^{-1/2}) = \delta + O(n^{-1/2}) = O(\delta)$, since $\delta\geq\eps\delta\geq\Omega(n^{-1/2})$.

    We now make cases based on whether $k$ is odd or even:
    \begin{enumerate}[(1)]
        \item $k$ is odd: In this case, \cref{fact:highcorrsos} implies that $\pE_\mu[\langle x, x^*\rangle]\geq n(1 - \xi)\implies\langle\pE_\mu[x], x^*\rangle\geq n(1 - 2\xi)$. Write $z:= \pE_\mu[x]$. Then note that $\|z\|_2^2 = \sum_i\pE_\mu[x_i]^2\overset{\text{\cref{sosbasicfacts}}}{\leq} n$. Consequently, $\corr(z, x^*) = \abs{\ip{z, x^*}}/\|z\|_2\cdot\|x^*\|_2\geq\langle z, x^*\rangle/n\geq 1 - 2\xi$. Now, by \cref{sosfkp}, the vector $x'$ we compute (in \cref{oddkprocedure} of \cref{alg:recxxor}) is very close to its ``idealized value'' $z$, i.e.\ $\|x' - z\|_1\leq 2^{-n}$. Since $2\xi\geq\Omega(n^{-1/2})\gg 2^{-n}$, $2\xi + O(2^{-n})\leq O(\xi) \leq O(\delta)$. 
        Now, by \cref{prop:realtoboolean}, it follows that for $\widehat{x} \defeq \sgn(x')$, we have that $\corr(\widehat{x}, x^*) \geq 1 - O(\delta)$, which finishes the proof in this case.

        \item $k$ is even: In this case, \cref{fact:highcorrsos} implies that $\pE_\mu[\langle x, x^*\rangle^2]\geq n^2(1 - 2\xi)$. Now, note that $\langle x, x^*\rangle^2 = \langle x^{\otimes 2}, x^{*\otimes 2}\rangle$, and thus $\langle\pE_\mu[x^{\otimes 2}], x^{*\otimes 2}\rangle\geq n^2(1 - 2\xi)$. Write $Z:= \pE_\mu[x^{\otimes 2}]$, where $Z$ is to be viewed as a $n\times n$ matrix. Then note that $\langle Z, x^{*\otimes 2}\rangle = \sum_{i\in[n]}\langle z^{(i)}, x^*_ix^*\rangle\geq n^2(1 - 2\xi)$, where $z^{(i)}$ is the $i^{\mathrm{th}}$ row of $Z$. Also note that $\|z^{(i)}\|_2^2 = \sum_j\pE_\mu\left[x_ix_j\right]^2\overset{\text{\cref{sosbasicfacts}}}{\leq} \sum_j \pE_\mu\left[x^2_ix^2_j\right] = n$, and thus
\begin{flalign*}
\sum_{i\in[n]}\abs{\corr(z^{(i)}, x^*)} &= \sum_{i\in[n]}\frac{|\langle z^{(i)}, x^*\rangle|}{\|z^{(i)}\|_2\cdot\|x^*\|_2} 
= \sum_{i\in[n]}\frac{|\langle z^{(i)}, x_i^*x^*\rangle|}{\|z^{(i)}\|_2\cdot\sqrt{n}} \\
&\geq \frac{1}{n}\sum_{i\in[n]}\abs{\langle z^{(i)}, x_i^*x^*\rangle}
\geq\frac{1}{n}\left|\sum_{i\in[n]}\langle z^{(i)}, x_i^*x^*\rangle\right|\geq n\left(1 - 2\xi\right)\mper
\end{flalign*}
On the other hand, $\sum_{i\in[n]}\abs{\corr(z^{(i)}, x^*)}\leq n$, since $\corr(z^{(i)}, x^*)\in[-1, 1]$ for each $i\in[n]$. Consequently, by an averaging argument, we must have that for $\geq 0.99n$ indices $i\in[n]$, $\abs{\corr(z^{(i)}, x^*)}\geq 1 - 200\xi$.

As before, if $X^{(i)}$ is the vector we compute (in \cref{evenkprocedure} of \cref{alg:approxrecovery}), then $\|X^{(i)} - z^{(i)}\|_1\leq 2^{-n}$, and thus for $\geq 0.99 n$ indices $i\in[n]$, we have $\abs{\corr(X^{(i)}, x^*)}\geq 1 - O(\xi) \geq 1 - O(\delta)$, and we can absorb the constant in $O(\cdot)$ into the lower bound for $m$, as usual.

Now, let $\widehat{x}^{(i)} = \sgn(X^{(i)})$ for each $i$. By \cref{prop:realtoboolean}, we have that $\abs{\corr(\widehat{x}^{(i)}, x^*)} \geq 1 - O(\delta)$ for at least $0.99 n$ indices $i \in [n]$. Call these indices the ``good'' indices. We now observe that the $i^*$ chosen in \cref{item:nearestdecode} must be such that $\abs{\corr(\widehat{x}^{(i^*)}, x^*)} \geq 1 - O(\delta)$. Indeed, for any choice of one of the ``good'' $0.99 n$ indices $i$, we must have that $\abs{\corr(\widehat{x}^{(i)}, \widehat{x}^{(j})} \geq 1 - O(\delta)$ for every other good $j$, i.e., $\delta_i = O(\delta)$. Therefore, for the particular $i^*$ chosen in \cref{item:nearestdecode}, it follows that $\abs{\corr(\widehat{x}^{(i^*)}, \widehat{x}^{(j})} \geq 1 - O(\delta)$ for some good index $j$, there are $ \geq 0.99 n$ good indices. We thus conclude by triangle inequality that $\abs{\corr(\widehat{x}^{(i^*)}, x^*)} \geq 1 - O(\delta)$, which finishes the proof.\qedhere
\end{enumerate}
\end{proof}

\subsubsection{Proof of \texorpdfstring{\cref{lem:localtoglobal}}{Lemma~\ref{lem:localtoglobal}}}
\label{sec:localtoglobal}
In this subsection, we prove \cref{lem:localtoglobal}, restated below.
\localtoglobal*

\begin{proof}
\label{refutecompletegraphproof}
Throughout what follows, we write $\supxe$ to mean $\sup_{\pE\in\SoS_{\geq 2\ell}}$. Now define 
\begin{equation*}f(\cH):= \supxe\pE\left[\E_{C\sim\cH}[x_C] - \E_{C\sim\cK}[x_C]\right]\mper\end{equation*}
Note that we need to show that $f(\cH)\leq\delta + O(n^{-1/2})$.

Now, since $\E_{C\sim\cK}[\cdot] = \E_{\cH'}\E_{C\sim\cH'}[\cdot]$, we have that
\begin{equation*}f(\cH) = \supxe\E_{\cH'}\left[\E_{C\sim\cH}\pE[x_C] - \E_{C\sim\cH'} \pE[x_C]\right]\mcom\end{equation*}
where $\E_{\cH'}$ takes expectations over hypergraphs which are $m$ u.a.r.\ samples of $\cK$.

However, 
\begin{equation*}\E_{\cH}f(\cH) = \E_{\cH}\supxe\E_{\cH'}\frac{1}{m}\sum_{i = 1}^m\pE(x_{A_i} - x_{A'_i})\leq\E_{\cH, \cH'}\supxe\frac{1}{m}\sum_{i = 1}^m\pE(x_{A_i} - x_{A'_i})\end{equation*} 
\begin{equation*}\overset{(\ast)}{=} \E_{\cH, \cH', \sigma}\supxe\frac{1}{m}\sum_{i = 1}^m\pE[\sigma_i(x_{A_i} - x_{A'_i})]\leq\E_{\cH, \cH', \sigma}\left[\supxe\frac{1}{m}\sum_{i = 1}^m\pE\sigma_ix_{A_i} + \supxe\frac{1}{m}\sum_{i = 1}^m\pE\left[-\sigma_ix_{A'_i}\right]\right]\end{equation*}
\begin{equation*} = 2\cdot\E_{\cH, \sigma}\supxe\frac{1}{m}\sum_{i = 1}^m\pE\sigma_ix_{A_i} = 2\cdot\E_{\cH, \sigma}\supxe\pE\E_{C\sim\cH}\sigma_Cx_C\end{equation*} 
\begin{align}
\label{fexpectationupperbound}
    \implies\E_{\cH}f(\cH)\leq 2\cdot\E_{\cH, \sigma}\supxe\pE\E_{C\sim\cH}\sigma_Cx_C\mper
\end{align}
Here, in $(\ast)$, $\sigma = (\sigma_i)_{i\in[m]}$ is a collection of $m$ i.i.d. Rademacher random variables, and $\cH = (A_1, \ldots, A_m), \cH' = (A'_1, \ldots, A'_m)$ enumerates the order in which $\cH, \cH'$ were sampled. Note that for any $i$, and any two clauses $C, C'\in\cK$, $\Pr(C = A_i, C' = A'_i) = \Pr(C' = A_i, C = A'_i)$, and thus we can introduce the random sign $\sigma_i$.

Now, by triangle inequality, 
\begin{align}
\label{ftriang}
    f(\cH)\leq|f(\cH) - \E_{\cH}f(\cH)| + |\E_{\cH}f(\cH)|\mper
\end{align}
Now, note that $|\E[\cdot]|\leq\E[|\cdot|], |\sup(\cdot)|\leq\sup(|\cdot|)$. Consequently, by \cref{fexpectationupperbound}, 
\begin{equation*}|\E_{\cH}f(\cH)|\leq 2\cdot\E_{\cH, \sigma}\supxe\left|\pE\E_{C\sim\cH}\sigma_Cx_C\right|\mper\end{equation*}
Now, by \cref{cor:refuterandomhypergraph}, for randomly chosen $\cH, \sigma$, we have that $\supxe\left|\pE\E_{C\sim\cH}\sigma_Cx_C\right|\leq \delta + 2k^2/n$ with probability $\geq 1 - 1/\poly(n)=: 1 - \gamma$. On the other hand, note that $\pE\E_{C\sim\cH}\sigma_Cx_C = \E_{C\sim\cH}\left[\sigma_C\pE[x_C]\right]$. By \cref{sosbasicfacts}, $\pE[x_C]\in[-1, 1]$ for any $C\in\cK$, and thus for all $\cH, \sigma$, we have that $\E_{C\sim\cH}[\sigma_C\pE[x_C]]\in[-1, 1]$, and thus $\supxe\left|\pE\E_{C\sim\cH}\sigma_Cx_C\right|\leq 1$. Consequently, 
\begin{equation*}\E_{\cH, \sigma}\supxe\left|\pE\E_{C\sim\cH}\sigma_Cx_C\right|\leq\left(\delta + \frac{2k^2}{n}\right)\cdot(1 - \gamma) + 1\cdot \gamma\leq \delta + \frac{2k^2}{n} +  \gamma\leq \delta + O\left(\frac{k^2}{n}\right)\mcom\end{equation*}
where the last inequality follows since $\gamma\leq 1/\poly(n)$, and we can choose constants appropriately such that $\gamma\ll 1/n$.

Consequently, $|\E_{\cH}f(\cH)|\leq 2\cdot\E_{\cH, \sigma}\supxe\left|\pE\E_{C\sim\cH}\sigma_Cx_C\right|\leq 2\delta + O(k^2/n)$. 

Now, let $\cH = (A_1, \ldots, A_i, \ldots, A_m)$, and let $\cG = (A_1, \ldots, A'_i,\ldots,A_m)$ be hypergraphs which differ only in the $i^{\mathrm{th}}$ hyperedge. Then
\begin{equation*}|f(\cH) - f(\cG)| = \left|\supxe\pE\left[\E_{C\sim\cH} x_C - \E_{C\sim\cK} x_C\right] - \supxe\pE\left[\E_{C\sim\cG} x_C - \E_{C\sim\cK} x_C\right]\right|\end{equation*} \begin{equation*}\overset{(**)}{\leq}\supxe\left|\pE\left[\E_{C\sim\cH} x_C - \E_{C\sim\cK} x_C\right] - \pE\left[\E_{C\sim\cG} x_C - \E_{C\sim\cK} x_C\right]\right| = \supxe\left|\pE\left[\E_{C\sim\cH} x_C - \E_{C\sim\cG} x_C\right]\right|\mper\end{equation*}
Here $(**)$ follows from the fact that $|\sup g_1 - \sup g_2|\leq\sup|g_1 - g_2|$ for any functions $g_1, g_2$.

Now, 
\begin{equation*}\E_{C\sim\cH} x_C - \E_{C\sim\cG} x_C = \frac{1}{m}\left(\sum_{j\in[m]}x_{A_j} - x_{A'_i} - \sum_{j\in[m]\setminus\{i\}}x_{A_j}\right) = \frac{1}{m}(x_{A_i} - x_{A'_i})\end{equation*}
\begin{equation*}\implies\pE\left[\E_{C\sim\cH} x_C - \E_{C\sim\cG} x_C\right] = \frac{1}{m}(\pE x_{A_i} - \pE x_{A'_i})\mper\end{equation*}
Now, for any pseudo-expectation $\pE$ and any $S\seq[n]$, we have that $\pE[x_S]\in[-1, 1]$. Consequently, $\frac{1}{m}(\pE x_{A_i} - \pE x_{A'_i})\in[-\frac{2}{m}, \frac{2}{m}]\implies|\frac{1}{m}(\pE x_{A_i} - \pE x_{A'_i})|\leq\frac{2}{m}$. Consequently, \begin{equation*}\supxe\left|\pE\left[\E_{C\sim\cH} x_C - \E_{C\sim\cG} x_C\right]\right|\leq\frac{2}{m}\mcom\end{equation*}
i.e.\ $|f(\cH) - f(\cG)|\leq\frac{2}{m}$ if $\cH, \cG$ differ in at most one entry. Consequently, by \cref{fact:mcdiarmid}, we have that $|f(\cH) - \E_{\cH}[f(\cH)]|\leq O(\sqrt{\log n/m})$ with probability $\geq 1 - 1/\poly(n)$.

Putting everything together, we have, by \cref{ftriang}, with probability $\geq 1 - 1/\poly(n)$ over the randomness of $\cH$,
\begin{equation*}f(\cH)\leq 2\delta + O\left(\frac{k^2}{n}\right) + O\left(\sqrt{\frac{\log n}{m}}\right)\leq 2\delta + O\left(\frac{1}{\sqrt{n}}\right)\mcom\end{equation*}
where the last inequality follows since $m\geq\Omega(n\log n)$ and $k\leq o(n^{1/4})$. Finally, we can absorb the $2$ in front of $\delta$ by increasing the $2^{O(k)}$ term in the hypothesis.
\end{proof}

\subsection{Rounding approximate solutions: proof of \texorpdfstring{\cref{lem:rounding}}{Lemma~\ref{lem:rounding}}}
\label{sec:rounding}

In this subsection, we prove \cref{lem:rounding}.

\begin{proof}

Since $|\cH|=\Omega(\eps^{-2}n\log n)$, and since $|[n]^k\setminus[n]^{(k)}|/n^k\leq O(k^2/n) = o(1)$ (see \cref{prop:hypergraphclean}), if we write $\cH':= \cH\cap[n]^{(k)}$, with probability $\geq 1 - \exp(-\Omega(n))$ we have $|\cH'|\geq \Omega(\eps^{-2}n\log n)$. Note that all tuples in $\cH'$ have distinct entries. For every $i\in[n]$, write $\cch:= \{C\setminus\{i\}: i\in C\in\cH'\}$. Note that the $\cch$s may have repeated hyperedges. Also for any $T\in\cch$, denote by $(T, i)$ the tuple in $\cH'$ that $T$ was obtained from. Since the hyperedges of $\cH'$ are independent uniform samples from $[n]^{(k)}$, the hyperedges of $\cch$ are independent uniform samples from $([n]\setminus\{i\})^{(k - 1)}$.

Consider the following algorithm:
\begin{mdframed}
    \begin{algorithm}[Exact Recovery from an Approximately Correct Solution]
    \label{alg:exactrecrand}
    \label{alg:rounding}
    \mbox{}
    \begin{description}
        \item[Input:] A $k$-XOR instance $(\cH, \{b_C\}_{C \in \cH})$, and a vector $\widetilde{x}\in\Fits^n$.
        \item[Output:] A vector $x\in\Fits^n$.
        \item[Operation:] \mbox{}
        \begin{enumerate}[(1)]
            \item For every $i\in[n]$, compute $S_i:= \{b_{(T, i)}\widetilde{x}_{T}:T\in\cch\}$, which is a multiset over $\Fits$. Here $\cch:= \{C\setminus\{i\}: i\in C\in\cH\cap[n]^{(k)}\}$.
            \item Output $x\in\Fits^n$, where 
                \begin{equation*}x_i:= \begin{cases}
                    1 & \text{if }\geq|S_i|/2\text{ entries of }S_i\text{ are }1\mcom\\
                    -1 & \text{otherwise}
                \end{cases}\mper\end{equation*}
        \end{enumerate}
    \end{description}
    \end{algorithm}
\end{mdframed}
It is easy to see that the above algorithm runs in $\poly(m,n)$ time. We will argue that if $k$ is odd, then $x = x^*$ with probability $\geq 1 - 1/\poly(n)$, and if $k$ is even, then $x \in \{x^*, -x^*\}$ with probability $1 - 1/\poly(n)$. 

Let us first consider the case where $\corr(\widehat{x},x^*) \geq 1 - \delta$ (for both odd and even $k$).
It suffices to show that $x_i\neq x^*_i$ with probability $\leq O(n^{-2})$ for all $i\in[n]$, as then we are done by a union bound over all $i\in[n]$.

Write $B:= \{i\in[n]: \widetilde{x}_i \neq x^*_i\}$. Note that $\langle \widetilde{x}, x^*\rangle = n - 2|B|$, and thus $|B|\leq\delta n/2$. We call $T\in\cch$ \emph{good} if it satisfies the following conditions:
    \begin{enumerate}[(1)]
        \item $T\cap B = \emptyset$, 
        \item $b_{(T, i)} = x^*_{(T, i)}$.
    \end{enumerate}
    Note that if $T\in\cch$ is good, then $b_{(T, i)}\widetilde{x}_{T} = x^*_{(T, i)}\widetilde{x}_{T} = x^*_{(T, i)}x^*_{T} = x^*_i$. Consequently, $x_i\neq x^*_i$ if $<|S_i|/2=|\cch|/2$ elements of $\cch$ are good. 
    
Since $\widetilde{x}$ is independent of $\cH$ and all $\cch$s, we see that the probability that any given element of $\cch$ does not intersect $B$ is 
\begin{equation*}\geq\prod_{i=1}^{k-1}\frac{n-|B|-i}{n-i}\geq\left(\frac{n - |B| - k}{n - k}\right)^{k - 1}\geq\left(\frac{1 - \frac{\delta}{2} - \frac{k}{n}}{1 - \frac{k}{n}}\right)^{k-1}\geq 1 - \frac{(k - 1)\delta}{2(1 - \frac{k}{n})}\geq 1 - \frac{k\delta}{2}\geq 1 - \frac{\eps}{2}\mcom\end{equation*}
where we recall $\cch$ contains independent uniformly random tuples from $([n]\setminus\{i\})^{(k - 1)}$, and $\frac{k-1}{1 - k/n}\leq k$ for large enough $n$.  Since the noise in the equations $\{x_C = b_C\}_{C\in\cH}$ is independent of the randomness in sampling $\cH$, the probability that any given element of $\cch$ is good is \begin{equation*}\geq\left(1-\frac{\eps}{2}\right)\left(\frac{1}{2}+\eps\right)\geq \frac{1}{2} + \frac{\eps}{4}\mper\end{equation*}
Now, by a Chernoff bound, $|\cch|\geq \frac{1}{2n}\cdot|\cH'|=\Omega(\eps^{-2}\log n)$ with probability \begin{equation*}\geq 1-\exp\left(-\Omega(|\cH'|/n)\right)=1-\exp\left(-\Omega(\eps^{-2}\log n)\right)\geq 1-1/\poly(n)\mper\end{equation*}
By another Chernoff bound, the probability that $\geq \frac{1}{2}$ fraction of $\cch$ is good is \begin{equation*}\geq 1- \exp(-\Omega(\eps^2|\cch|))\geq 1-\exp(-\Omega(\log n))\geq 1-1/\poly(n)\mcom\end{equation*}
as desired.

It remains to handle the case when $k$ is even and $\corr(-\widehat{x}, x^*) \geq 1 - \delta$. Let us make the following simple observation: since $k$ is even, the distributions $\LPN_k(x^*, m, \eps)$ and $\LPN_k(-x^*, m, \eps)$ are identical. Hence, we can follow the above analysis using $-x^*$ as the planted assignment since $\corr(\widehat{x}, -x^*) \geq 1 - \delta$, and the aforementioned argument implies that we recover $-x^*$ with probability at least $1 - 1/\poly(n)$, which finishes the proof.
\end{proof}

\bibliographystyle{alpha}
\bibliography{main.bbl}
\appendix

\section{Proofs of Auxilliary Statements}
\label{app:sosrounding}
In this appendix, we prove \cref{fact:highcorrsos,cor:refuterandomhypergraph}.

\subsection{Proof of \texorpdfstring{\cref{fact:highcorrsos}}{Fact~\ref{fact:highcorrsos}}}
\label{sec:highcorrsos}
In this subsection, we prove \cref{fact:highcorrsos}, which is a generalization of \cite[Lemma~A.5]{HopkinsSS15} (which holds for $k = 3$) for all $k$. It says that if $\pE[\langle x, x^*\rangle^k]$ is large for a pseudo-expectation $\pE$ of degree $\geq k + 2$, then $\pE[x]$ (resp.\ $\pE[x^{\otimes 2}]$) is strongly correlated with $x^*$ (resp.\ $x^{*\otimes 2}$) for odd (resp.\ even) $k$.
\highcorrsos*
\begin{proof}
\label{highcorrsosproof}

We split into cases based on the parity of $k$:

 \parhead{Case 1: $k$ is odd.} Let $p_k\in\R[u]$ such that $p_k(u)= 1-2(u/n)^k+u/n $. Note that $p_k(u)\geq 0$ for all $u\in [-n,n]$. By \cite[Fact~3.2]{OZhou13}, there exist SoS polynomials $s_0, s_1, s_2$ such that $\deg(s_0)\leq k, \deg(s_1), \deg(s_2)\leq k - 1$, and $p_k(u)=s_0(u)+s_1(u)(n+u)+s_2(u)(n-u).$

 Since $s_0$ is a SoS polynomial, $\pE[s_0]\geq 0$, and thus we have
    \begin{align*}
        \pE[ p_k(\langle x, x^* \rangle)]\geq \pE \left[s_1(\langle x, x^* \rangle)(n+\langle x, x^* \rangle)\right]+\pE \left[s_2(\langle x, x^* \rangle)(n-\langle x, x^* \rangle)\right]\mper
    \end{align*}
    Now, note that $\pE[s_1(x)\|x + x^*\|_2^2]\geq 0$, since $s_1$ is a SoS polynomial of degree $\leq k - 1$, $\|x + x^*\|_2^2 = \sum_{i = 1}^n(x_i + x^*_i)^2$ is a SoS polynomial of degree $2$, and $\pE$ is a pseudo-expectation of degree $\geq k + 1$. Consequently, 
    \begin{equation*}\pE[s_1(x)\|x + x^*\|_2^2]\geq 0\Longleftrightarrow\pE\left[s_1(x)\left(\langle x, x^*\rangle + \frac{1}{2}(\|x\|_2^2 + \|x^*\|_2^2)\right)\right]\geq 0\end{equation*}
    \begin{equation*}\Longleftrightarrow \pE \left[s_1(\langle x, x^* \rangle)\langle x, x^* \rangle\right]\geq-\frac{1}{2}\pE[s_1(\langle x, x^*\rangle)(\lVert x\rVert^2+n)] = -n\pE[s_1(\langle x, x^*\rangle)]\mcom\end{equation*}
    where the last equality follows from \Cref{sosbasicfacts}. Therefore, we have
    \begin{align*}
        \pE \left[s_1(\langle x, x^* \rangle)(n+\langle x, x^* \rangle)\right]
        &\geq n\pE[s_1(\langle x, x^*\rangle)]-n\pE[s_1(\langle x, x^*\rangle)] = 0\mcom
    \end{align*}
    Similarly, using the fact $\pE[s_2(x)\|x - x^*\|_2^2]\geq 0$, we have
 \begin{align*}
        \pE \left[s_2(\langle x, x^* \rangle)(n-\langle x, x^* \rangle)\right]
        &\geq n\pE[s_2(\langle x, x^*\rangle)]-n\pE[s_2(\langle x, x^*\rangle)]=0\mper
    \end{align*}
    Thus by combining these inequalities, we get $\pE [p_k(\langle x, x^*\rangle)]\geq 0$. This is equivalent to 
    \begin{align*}
        \pE \langle x, x^* \rangle\geq n\left( 2\frac{\pE\langle x, x^* \rangle^k}{n^k} - 1\right)\geq n(2(1-\delta)-1) = n(1-2\delta)\mcom
    \end{align*}
    as desired.

    \parhead{Case 2: $k$ is even.} 
    Let $p_k\in\R[u]$ such that $p_k(u)=1-2(u/n)^k+(u/n)^2$. Note that $p_k(u)\geq 0$ for all $u\in [-n,n]$. By \cite[Fact~3.2]{OZhou13}, there exist SoS polynomials $s_0, s_1$ such that $\deg(s_0)\leq k, \deg(s_1)\leq k - 2$, and $p_k(u)=s_0(u)+s_1(u)(n+u)(n-u)$.
    
    In this case, we have 
    \begin{align*}
        \pE[p_k(\langle x, x^* \rangle)]\geq \pE \left[s_1(\langle x, x^* \rangle)(n^2 - \langle x, x^* \rangle^2)\right]\mper
    \end{align*}
    However, note that 
    \begin{equation*}\langle x, x^*\rangle^2 = \left(\sum_{i = 1}^n x_ix^*_i\right)^2 = n + 2\sum_{1\leq i < j\leq n}x_ix_jx^*_ix^*_j = n + \sum_{1\leq i < j\leq n}\left(2 - (x_ix^*_j - x^*_ix_j)^2\right)\end{equation*}
    \begin{equation*}= n + 2\binom{n}{2} - \sum_{1\leq i < j\leq n}(x_ix^*_j - x^*_ix_j)^2\implies n^2 - \langle x, x^*\rangle^2 = \sum_{1\leq i < j\leq n}(x_ix^*_j - x^*_ix_j)^2\mper\end{equation*}
    Consequently, $n^2 - \langle x, x^*\rangle^2$ is a degree $4$ SoS polynomial on $\Fits^n$. Thus, by \cref{sosaxioms}, we have $\pE \left[s_1(\langle x, x^* \rangle)(n^2 - \langle x, x^* \rangle^2)\right]\geq 0$,\footnote{Formally, one would write $n^2 - \langle x, x^*\rangle^2 = \sum_{1\leq i < j\leq n}(x_ix^*_j - x_jx^*_i)^2 + n^2 - \sum_{i\in[n]}(x_ix^*_i)^2 - \sum_{1\leq i < j\leq n}(x_i^2x^{*2}_j + x_j^2x^{*2}_i)$. Since $\pE[fx_i^2] = \pE[f]$ for any function $f$ by \cref{sosaxioms}, we can then get rid of the $n^2 - \sum_{i\in[n]}(x_ix^*_i)^2 - \sum_{1\leq i < j\leq n}(x_i^2x^{*2}_j + x_j^2x^{*2}_i)$ term.} since $\deg(s_1)\leq k - 2$, the degree of $n^2 - \langle x, x^*\rangle^2$ is $4$, and $\pE$ is a pseudo-expectation of degree $\geq k + 2$. Thus $\pE \left[p_k(\langle x, x^* \rangle)\right]\geq 0$. But then plugging in the definition of $p_k(u)$ we obtain 
    \begin{align*}
        \pE \langle x, x^* \rangle^2\geq n^2\left(2\frac{\pE\langle x, x^* \rangle^k}{n^k} - 1\right)\geq n^2(2(1-\delta)-1) = n^2(1-2\delta)\mcom
    \end{align*}
    as desired.
\end{proof}

\subsection{Proof of \texorpdfstring{\cref{cor:refuterandomhypergraph}}{Corollary~\ref{cor:refuterandomhypergraph}}}

In this subsection, we prove \cref{cor:refuterandomhypergraph}, which is \cref{fact:hkmrefute} except we allow the hyperedges of $\cH$ to be drawn from $[n]^k$ instead of $[n]^{(k)}$, i.e., they need not have all distinct entries. This is a straightforward corollary of applying \cref{prop:hypergraphclean}.
\refuterandomhypergraph*
\begin{proof}
\label{refuterandomhypergraphproof}
    Write $\cH':= \cH\cap[n]^{(k)}$. By \cref{prop:hypergraphclean}, $|\cH\setminus\cH'|\leq 2mk^2/n$ with probability $\geq 1 - \exp(-\Omega(mk^2/n))$. Now, note that $m\geq 2^{O(k)}n\log n$ since $\ell\leq n$. Consequently, $\exp(-\Omega(mk^2/n))\leq 1/\poly(n)$.
    
    Now, $\pE\sum_{C\in\cH}\sigma_Cx_C = \pE\sum_{C\in\cH'}\sigma_Cx_C + \pE\sum_{C\in\cH\setminus\cH'}\sigma_Cx_C\leq \pE\sum_{C\in\cH'}\sigma_Cx_C + |\cH\setminus\cH'|$.\footnote{Here, we use the fact that $\pE\sum_{C\in\cH\setminus\cH'}\sigma_Cx_C = \sum_{C\in\cH\setminus\cH'}\pE[\sigma_Cx_C]$. Since $(\sigma_Cx_C)^2 = 1$, by \cref{sosbasicfacts}, $\pE[\sigma_Cx_C]^2\leq\pE[(\sigma_Cx_C)^2] = \pE[1] = 1\implies|\pE[\sigma_Cx_C]|\leq 1$ for all pseudo-expectations.} On the other hand, by \cref{fact:hkmrefute},\footnote{Since we're applying \cref{fact:hkmrefute} to $\cH'$, technically we should have $|\cH'|$ in the lower bound instead of $m$, but we can absorb the extra factor of $1/(1 - 2k^2/n)$ in the $2^{O(k)}$ factor.} we have $|\pE[\sum_{C\in\cH'}\sigma_Cx_C]|\leq\delta|\cH'|\leq\delta m$. Consequently, 
    \begin{equation*}\left|\pE\E_{C\sim\cH}\sigma_Cx_C\right|\leq\frac{1}{m}\left|\pE\sum_{C\in\cH'}\sigma_Cx_C\right| + \frac{|\cH\setminus\cH'|}{m}\leq\delta + \frac{2k^2}{n}\mper\qedhere\end{equation*}
\end{proof}

\end{document}